\documentclass[journal,twocolumn]{IEEEtran}
\usepackage{cite}
\usepackage[cmex10]{amsmath}
\usepackage{amssymb, amsthm}
\usepackage{subeqn}
\usepackage{graphicx}
\usepackage{xcolor}
\usepackage[font=footnotesize]{subfig}
\usepackage{mathrsfs}
\usepackage{dsfont}

\newcommand{\ignore}[1]{}

\newtheorem{prob-statement}{Problem}

\newtheorem{thrm}{Theorem}

\DeclareMathOperator*{\minimize}{minimize}

\DeclareMathOperator*{\var}{Var}

\newenvironment{bluetext}{\par\color{black}}{\par}

\begin{document}

\title{Distributed Inference with M-ary Quantized Data in the Presence of Byzantine Attacks}

\author{V.~Sriram~Siddhardh~(Sid)~Nadendla,~\IEEEmembership{Student Member,~IEEE,}
        Yunghsiang~S.~Han,~\IEEEmembership{Fellow,~IEEE,}
        and~Pramod~K.~Varshney,~\IEEEmembership{Fellow,~IEEE}

\thanks{V. Sriram Siddhardh (Sid) Nadendla and Pramod K. Varshney are with the Department
of Electrical Engineering and Computer Science, Syracuse University, Syracuse, NY 13201, USA. E-mail: \{vnadendl, varshney\}@syr.edu.}

\thanks{Yunghsiang S. Han is with the Department of Electrical Engineering, National Taiwan University of Science and Technology, Taipei, Taiwan. E-mail: yshan@mail.ntust.edu.tw‎.}

}


\maketitle


\begin{abstract}
The problem of distributed inference with M-ary quantized data at the sensors is investigated in the presence of Byzantine attacks. We assume that the attacker does not have knowledge about either the true state of the phenomenon of interest, or the quantization thresholds used at the sensors. Therefore, the Byzantine nodes attack the inference network by modifying modifying the symbol corresponding to the quantized data to one of the other M symbols in the quantization alphabet-set and transmitting the false symbol to the fusion center (FC). In this paper, we find the optimal Byzantine attack that \emph{blinds} any distributed inference network. As the quantization alphabet size increases, a tremendous improvement in the security performance of the distributed inference network is observed. 

We also investigate the problem of distributed inference in the presence of resource-constrained Byzantine attacks. In particular, we focus our attention on two problems: distributed detection and distributed estimation, when the Byzantine attacker employs a highly-symmetric attack. For both the problems, we find the optimal attack strategies employed by the attacker to maximally degrade the performance of the inference network. A reputation-based scheme for identifying malicious nodes is also presented as the network's strategy to mitigate the impact of Byzantine threats on the inference performance of the distributed sensor network.
\end{abstract}




\IEEEpeerreviewmaketitle

\section{Introduction \label{sec: Introduction}}

Distributed inference in sensor networks has been widely studied by several scholars in the past three decades (See \cite{Book-Varshney, Akyildiz2002, Veeravalli2012, Chapter-Tsitsiklis1993, Viswanathan1997, Blum1997, Chen2006, Chamberland2007, Cheng2012, Patwari2005, Brooks2003} and references therein). The distributed inference framework comprises of a group of spatially distributed sensors which acquire observations about a phenomenon of interest (POI) and send processed data to a fusion center (FC) where a global inference is made. Due to resource-constraints in sensor networks, this data is processed at the sensors in such a way that the observations are mapped to symbols from an alphabet set of size M, prior to transmission to the FC.  When $M = 2$, we employ binary quantization to generate processed data. When $M > 2$, we send an M-ary symbol that is assumed to be generated via fine quantization. A sensor decision rule is assumed to be characterized by a set of quantization thresholds. In this paper, we use the phrases `\emph{mapped to one of the M-ary symbols}' and `\emph{quantized to an M-ary symbol}' interchangeably. A lot of work in the past has focussed on the binary quantization case, i.e., $M = 2$. In this paper, we consider the case of more general $M$, $M = 2$ being a special case. The framework of distributed inference networks has been extensively studied for different types of inference problems such as detection (e.g., \cite{Book-Varshney, Tsitsiklis1988, Viswanathan1997, Blum1997, Chen2006, Chamberland2007, Veeravalli2012}), estimation (e.g., \cite{Cheng2012, Patwari2005, Veeravalli2012}), and tracking (e.g., \cite{Brooks2003, Veeravalli2012}) in the presence of both ideal and non-ideal channels. In this paper, we focus our attention on two distributed inference problems, namely \emph{detection} and \emph{estimation} in the framework of distributed inference\ignore{ with ideal channels}, where sensors quantize their data to M-ary symbols. 

Although the area of sensor networks has been a very active field of research in the past, security problems in sensor networks have gained attention only \textcolor{black}{in the last decade \cite{Perrig2002, Perrig2004, Karlof2004}}. As the security threats have evolved more specifically directed towards inference networks, attempts have been made at the system-level to either prevent or mitigate these threats from deteriorating the network performance. While there are many types of security threats, in this paper, we address the problem of one such attack, called the Byzantine attack, in the context of distributed inference networks (see a recent survey \cite{Vempaty2013b} by Vempaty \emph{et al.}). \textcolor{black}{Byzantine attacks (proposed by Lamport \emph{et al}. in \cite{Lamport1982}) in general, are arbitrary and may refer to many types of malicious behavior. In this paper, we focus only on the data-falsification aspect of the Byzantine attack wherein one or more compromised nodes of the network send false information to the FC in order to deteriorate the inference performance of the network. A well known example of this attack is the \emph{man-in-the-middle} attack \cite{Nayak2010} where, on one hand, the attacker collects data from the sensors whose authentication process is compromised by the attacker emulating as the FC, while, on the other hand, the attacker sends false information to the FC using the compromised sensors' identity. In summary, if the $i^{th}$ sensor's authentication is compromised, the attacker remains invisible to the network, accepts the true decision $u_i$ from the $i^{th}$ sensor and sends $v_i$ to the FC in order to deteriorate the inference performance.}

Marano \emph{et al.}, in \cite{Marano2009}, analyzed the Byzantine attack on a network of sensors carrying out the task of distributed detection, where the attacker is assumed to have complete knowledge about the hypotheses. This represents the extreme case of Byzantine nodes having an extra power of knowing the true hypothesis. In their model, they assumed that the sensors quantized their respective observations into M-ary symbols, which are later fused at the FC. The Byzantine nodes pick symbols using an optimal probability distribution that are conditioned on the true hypotheses, and transmit them to the FC in order to maximally degrade the detection performance. Rawat \emph{et al.}, in \cite{Rawat2011}, also considered the problem of distributed detection in the presence of Byzantine attacks with binary quantizers at the sensors in their analysis. Unlike the authors in \cite{Marano2009}, Rawat \emph{et al.} did not assume complete knowledge of the true hypotheses at the Byzantine attacker. Instead, they assumed that the Byzantine nodes derive the knowledge about the true hypotheses from their own sensing observations. In other words, a Byzantine node potentially flips the local decision made at the node. It does not modify the thresholds at the sensor quantizers. Rawat \emph{et al.} also analyzed the performance of the network in the presence of independent and collaborative Byzantine attacks and modeled the problem as a zero-sum game between the sensor network and the Byzantine attacker. In addition to the analysis of distributed detection in the presence of Byzantine attacks, a reputation-based scheme was proposed by Rawat \emph{et al.} in \cite{Rawat2011} for identifying the Byzantine nodes by accumulating the deviations between each sensor's decision and the FC's decision over a time window of duration $T$. If the accumulated number of deviations is greater than a prescribed threshold for a given node, then the FC tags it as a Byzantine node. In order to mitigate the attack, the FC removes nodes which are tagged Byzantine node from the fusion rule. Another mitigation scheme was proposed by Vempaty \emph{et al.} \cite{Vempaty2011}, where each sensor's behavior is learnt over time and compared to the known behavior of the honest nodes. Any significant deviation in the learnt behavior from the expected honest behavior is labelled Byzantine node. Having learnt their parameters, the authors also proposed the use of this information to adapt their fusion rule so as to maximize the performance of the FC. In contrast to the parallel topology in sensor networks, Kailkhura \emph{et al.} in \cite{Kailkhura2013} investigated the problem of Byzantine attacks on distributed detection in a hierarchical sensor network. They presented the optimal Byzantine strategy when the sensors communicate their decisions to the FC in multiple hops of a balanced tree. They assumed that the cost of compromising sensors at different levels of the tree varies, and found the optimal Byzantine strategy that minimizes the cost of attacking a given hierarchical network.

Soltanmohammadi \emph{et al.} in \cite{Soltanmohammadi2013} investigated the problem of distributed detection in the presence of different types of Byzantine nodes. Each Byzantine node type corresponds to a different operating point, and, therefore, the authors considered the problem of identifying different Byzantine nodes, along with their operating points. The problem of maximum-likelihood (ML) estimation of the operating points was formulated and solved using the expectation-maximization (EM) algorithm. Once the Byzantine node operating points are estimated, this information was utilized at the FC to mitigate the malicious activity in the network, and also to improve global detection performance.

Distributed target localization in the presence of Byzantine attacks was addressed by Vempaty \emph{et al.} in \cite{Vempaty2013}, where the sensors quantize their observations into binary decisions, which are transmitted to the FC. Similar to Rawat \emph{et al.}'s approach in \cite{Rawat2011}, the authors in \cite{Vempaty2013} investigated the problem of distributed target localization from both the network's and Byzantine attacker's perspectives, first by identifying the optimal Byzantine attack and second, mitigating the impact of the attack with the use of non-identical quantizers at the sensors. 

In this paper, we extend the framework of Byzantine attacks when Byzantine nodes do not have complete knowledge about the true state of the phenomenon-of-interest (POI), and when the sensors generate M-ary symbols instead of binary symbols. We also assume that the Byzantine attacker is ignorant about the quantization thresholds used at the sensors to generate the M-ary symbols.\footnote{The well-known attacker-in-the-middle is one such example.} Under these assumptions, we address two inference problems: binary hypotheses-testing and parameter estimation. 

The main contributions of the paper are three-fold. First, we define a Byzantine attack model for a sensor network with individual sensors quantizing their observations into one of the M-ary symbols, when the attacker does not have complete knowledge about the true state of the POI and thresholds employed by the sensors. We model the attack strategy as a flipping probability matrix, where $(i,j)^{th}$ entry represents the probability with which the $i^{th}$ symbol is flipped into the $j^{th}$ symbol. Second, we show that quantization into M-ary symbols at the sensors, as opposed to binary quantization, improves both inference as well as security performance simultaneously. As a function of the number of Byzantine nodes in the network, we derive the optimal flipping matrix. Finally, we extend the mitigation scheme presented by Rawat \emph{et al.} in \cite{Rawat2011} to the more general case where sensors generate M-ary symbols. We present simulation results to illustrate the performance of the reputation-based scheme for the identification of Byzantine nodes in the network. 

The remainder of the paper is organized as follows. In Section \ref{sec: Model}, we describe our system model and present the Byzantine attack model for the case where sensors generate M-ary symbols when the attacker has no knowledge about the true state of the phenomenon of interest and quantization thresholds employed by the sensors. Next, we determine the most powerful attack strategy that the Byzantine nodes can adopt in Section \ref{sec: Analysis}. In the case of resource-constrained Byzantine attacks, where the attacker cannot compromise enough number of nodes in the network to \emph{blind} it (to be defined in Section \ref{sec: Model}), we find the optimal Byzantine attack for a fixed fraction of Byzantine nodes in the network in the context of distributed detection and estimation in Sections \ref{sec: Detection} and \ref{sec: Estimation} respectively. From the network's perspective, we present a mitigation scheme in Section \ref{sec: Reputation - Byzantine Identification} that identifies the Byzantine nodes using reputation-tags. Finally, we present our concluding remarks in Section \ref{sec: Conclusion}.

\section{System Model \label{sec: Model}}

\begin{figure*}[!t]
	\centerline
    {
    	\subfloat[Sensor Network Model]{\includegraphics[width=3.3in]{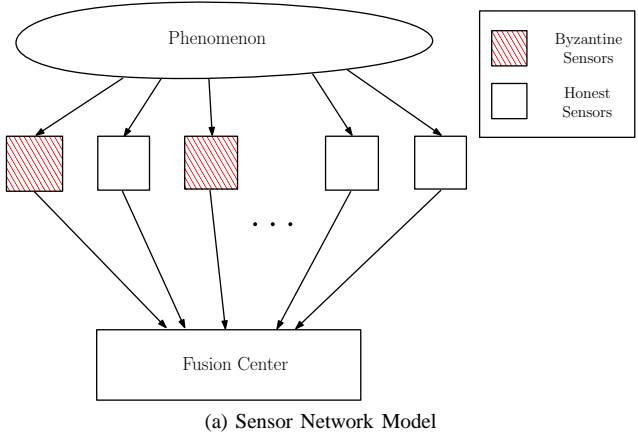}%
        \label{Fig: model}}
        \hfil
        \subfloat[Byzantine Attack Model]{\includegraphics[width=3.3in]{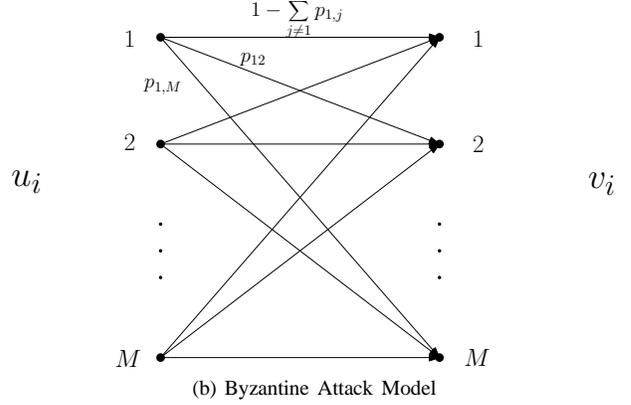}%
        \label{Fig: byzantine}}
    }
    \caption{Distributed Inference Network in the Presence of Byzantine Attacks}
    \label{Fig: Network Model}
\end{figure*}

Consider an inference (sensor) network with N sensors, where $\alpha$ fraction of the nodes in the network are assumed to be compromised (Refer to Figure \ref{Fig: model}). These compromised sensors transmit false data to the fusion center (FC) in order to deteriorate the inference performance of the network. We assume that the network is designed to infer about a particular phenomenon, regarding which sensors acquire \emph{conditionally-independent} observations. We denote the observation of the $i^{th}$ sensor as $r_i$. This observation $r_i$ is mapped to one of the $M$ symbols, $u_i \in \{ 1, \cdots, M \}$. In a compromised inference network, since the Byzantine sensors do not transmit their true quantized data, we denote the transmitted symbol as $v_i$ at the $i^{th}$ sensor. If the node $i$ is honest, then $v_i = u_i$. Otherwise, we assume that the Byzantine sensor modifies $u_i = l$ to $v_i = m$ with a probability $p_{lm}$, as shown in Figure \ref{Fig: byzantine}. For the sake of compactness, we denote the transition probabilities depicted in the graph in Figure \ref{Fig: byzantine} using a row-stochastic matrix $\mathbb{P}$, as follows:
\begin{equation}
	\mathbb{P} = 
	\left[ 
	\begin{matrix}
		p_{11} &  p_{12} & \ldots & p_{1M}\\
		p_{21} &  p_{22} & \ldots & p_{2M}\\
		\vdots & \vdots  & \ddots & \vdots\\
		p_{M1} &  p_{M2} & \ldots & p_{MM}
	\end{matrix}
	\right].
\end{equation}

Since the attacker has no knowledge of quantization thresholds employed at each sensor, we assume that $\mathbb{P}$ is independent of the sensor observations. The messages $\mathbf{v} = \{ v_1, v_2, \cdots, v_N \}$ are transmitted to the fusion center (FC) where a global inference is made about the phenomenon of interest based on $\mathbf{v}$. 

In order to consider the general inference problem, we assume that $\theta \in \Theta$ is the parameter that denotes the phenomenon of interest in the received signal $r_i$ at the $i^{th}$ sensor. If we are considering a detection/classification problem, $\theta$ is discrete (finite or countably infinite). In the case of parameter estimation, $\Theta$ is a continuous set. Without any loss of generality, we assume $\Theta = \{ 0, 1, \cdots, K-1 \}$ if the problem of interest is classification. Hence, detection is a special case of classification with $K = 2$. In the case of estimation, we assume that $\Theta = \mathbb{R}$.

Note that the performance of the FC is determined by the probability distribution (mass function) $P(\mathbf{v}|\theta)$. Therefore, in Section \ref{sec: Analysis}, we analyze the behavior of $P(\mathbf{v}|\theta)$ in the presence of different attacks and identify the one with the greatest impact on the network.
        
\section{Optimal Byzantine Attacks \label{sec: Analysis}}

Given the conditional distribution of $r_i$, $p(r_i | \theta)$, and the sensor quantization thresholds, $\lambda_j$ for $0\le j\le M$, the conditional distribution of $u_i$ can be found as
\begin{equation}
	P(u_i = m|\theta) = \displaystyle \int_{\lambda_{m-1}}^{\lambda_m} p(r_i|\theta) dr_i
\end{equation}
for all $m = 1, 2, \cdots, M$.

If the true quantized symbol at the $i^{th}$ node is $u_i = m$, a compromised node will modify it into $v_i = l$ as depicted in Figure \ref{Fig: byzantine}, and transmit it to the FC. Since the FC is not aware of the type of the node (honest or Byzantine), it is natural to assume that node $i$ is compromised with probability $\alpha$, where $\alpha$ is the fraction of nodes in the network that are compromised. Therefore, we find the conditional distribution of $v_i$ at the FC as follows.

\begin{equation}
	\begin{array}{l}
		P(v_i = m|\theta) \ = \ \alpha P(v_i = m|i = Byzantine, \theta) 
		\\ \\ \qquad \qquad \qquad + (1 - \alpha) P(v_i = m|i = Honest, \theta)
		\\
		\\
		\qquad = \displaystyle \alpha \sum_{l = 1}^{M} P(u_i = l|\theta) \cdot P(v_i = m|u_i = l, \theta) 
		\\ \\ \qquad \qquad \qquad + (1 - \alpha) P(u_i = m|\theta)
		\\
		\\
		\qquad = \displaystyle \alpha \sum_{l = 1}^{M} p_{lm} P(u_i = l|\theta) + (1 - \alpha) P(u_i = m|\theta)
		\\
		\\
		\qquad = \displaystyle \alpha \sum_{l \neq m} p_{lm} P(u_i = l|\theta) + [(1 - \alpha) + \alpha p_{mm}] P(u_i = m|\theta)
		\\
		\\
		\qquad = \displaystyle [(1 - \alpha) + \alpha p_{mm}]
		\\ \\ \qquad \qquad \qquad + \displaystyle \sum_{l \neq m}\{\alpha p_{lm} -[(1 - \alpha) + \alpha p_{mm}]\} P(u_i = l|\theta).
	\end{array}
	\label{Eqn: Cond. Prob vi given theta}
\end{equation}
The goal of a Byzantine attack is to blind the FC with the least amount of effort (minimum $\alpha$). To totally blind the FC is equivalent to making $P(v_i = m|\theta)=1/M$ for all $0\le m\le M-1$. In Equation (\ref{Eqn: Cond. Prob vi given theta}), the RHS consists of two terms. The first one is based on prior knowledge and the second term conveys information based on the observations. In order to blind the FC, the attacker should make the second term equal to zero. Since the attacker does not have any knowledge regarding $P(u_i = l|\theta)$, it can make the second term of Equation (\ref{Eqn: Cond. Prob vi given theta}) equal to zero by setting
\begin{equation}
	\alpha p_{lm} = (1 - \alpha) + \alpha p_{mm}, \qquad \forall\ l \neq m.
	\label{Eqn: Blinding condition}
\end{equation} 
Then the conditional probability $P(v_i=m|\theta)=(1-\alpha)+\alpha p_{mm}$ becomes independent of the observations $r_i$ (or its quantized version $u_i$), resulting in equiprobable symbols at the FC. In other words, the received vector $\mathbf{v} = \{ v_1, v_2, \cdots, v_N \}$ does not carry any information about $\theta$ and, therefore, results in the most degraded performance at the FC. So, the FC now has to solely depend on its prior information about $\theta$ in making an inference.

Having identified the condition in Equation (\ref{Eqn: Blinding condition}) under which the Byzantine attack makes the greatest impact on the performance of the network, we identify the strategy that the attacker should employ in order to achieve this condition as follows.  Since we need
$$P(v_i=m|\theta)=(1-\alpha)+\alpha p_{mm}=1/M,$$
 $\alpha=\frac{M-1}{(1-p_{mm})M}$.
To minimize $\alpha$, one needs to make $p_{mm}=0$. 
In this paper, we denote the $\alpha$ corresponding to this optimal strategy that minimizes the Byzantine attacker's resources required to blind the FC as $\alpha_{blind}$. Hence,
$$\alpha_{blind}=\frac{M-1}{M}.$$
Rearranging Equation (\ref{Eqn: Blinding condition}), we have
\begin{equation}
	\displaystyle \frac{1}{\alpha} = 1 + (p_{lm} - p_{mm})=1+p_{lm} \qquad \forall\ l \neq m.
	\label{Eqn: Blinding condition 2}
\end{equation}
By setting $\alpha$ to $\alpha_{blind}$, we have
$p_{lm}=1/(M-1)$, $\forall\ l \neq m$. That is, the transition probability $\mathbb{P}$ is a highly-symmetric matrix. We summarize the result as a theorem as follows.
\begin{thrm}
\label{thrm: Optimal Attack}
If the Byzantine attacker has no knowledge of the quantization thresholds employed at each sensor, then the optimal Byzantine attack is given as 
\begin{equation}
	\begin{array}{lcl}
		p_{lm} & = & \begin{cases} \displaystyle \frac{1}{M-1} & ; \text{ if } l \neq m \\ 0 &; \text{ otherwise} \end{cases}
		\\
		\\
		\alpha_{blind} & = & \displaystyle \frac{M-1}{M}.
		\\
	\end{array}
	\label{Eqn: Optimal Attack}
\end{equation}
\end{thrm}

We term Equation (\ref{Eqn: Optimal Attack}) as the optimal Byzantine attack, since the FC does not get any information from the data $\mathbf{v}$ it receives from the sensors to perform an inference task. Therefore, the FC has to rely on prior information about the parameter $\theta$, if available. 
\begin{bluetext}Theorem~\ref{thrm: Optimal Attack} can be extended to the case where the channels between sensors (attackers) are not perfect. The result is given in Appendix~\ref{sec:non-perfect channel}.\end{bluetext}

\begin{figure*}[!t]
	\centerline
    {
    	\subfloat[$\alpha_{blind}$ vs. $M$]{\includegraphics[width=4in]{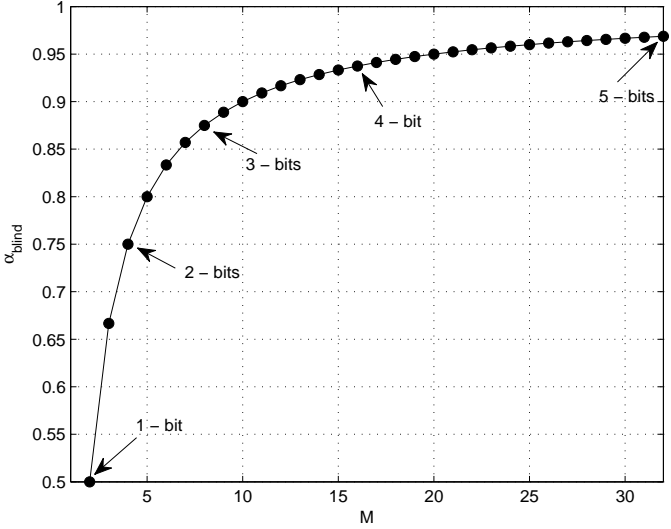}%
        \label{Fig: alpha_vs_M}}
        \hfil
        \subfloat[$\alpha_{blind}$ vs. Number of quantization bits ($\log_2 M$)]{\includegraphics[width=4in]{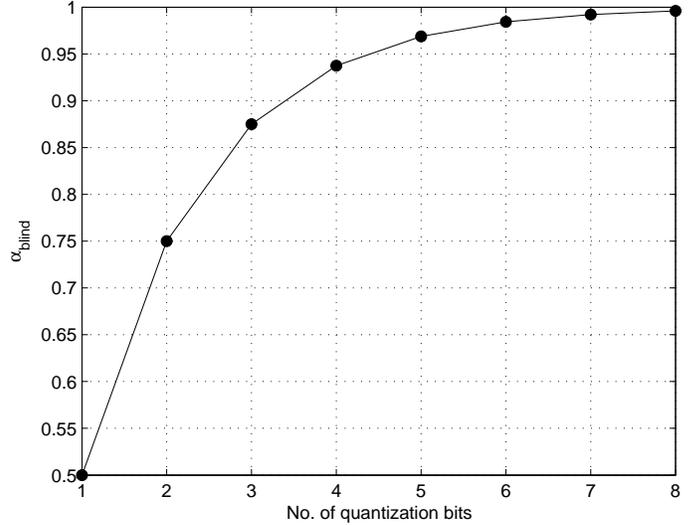}%
        \label{Fig: alpha_vs_bits}}
    }
    \caption{Improvement in $\alpha_{blind}$ with increasing number of quantization levels}
    \label{Fig: alpha_vs_M_bits}
\end{figure*}


\begin{table}[!t]
	\begin{center}
		\begin{tabular}{|c|c|}
			\hline
			Quantization bits & $\alpha_{blind}$
			\\
			\hline \hline
			1 & 0.5 
			\\
			\hline
			2 & 0.75
			\\
			\hline
			3 & 0.875
			\\
			\hline
			4 & 0.9375
			\\
			\hline
			5 & 0.9688
			\\
			\hline
			6 & 0.9844
			\\
			\hline
			7 & 0.9922
			\\
			\hline
			8 & 0.9961
			\\
			\hline
		\end{tabular}
		\caption{Improvement in $\alpha_{blind}$ with increasing number of quantization bits, $\log_2 M$}
		\label{Table: alpha_blind_bits}
	\end{center}
\end{table}

In Figure \ref{Fig: alpha_vs_M_bits}, we show how $\alpha_{blind}$ scales with increasing quantization alphabet size, $M$. Since the quantized symbols are encoded into bits, we also show an exponential increase in $\alpha_{blind}$ as the number of bits needed to encode the $M$ symbols, i.e., $\log_2 M$, increases. This is also shown in Table \ref{Table: alpha_blind_bits}. Note that, if the sensors use one additional quantization-bit (2-bit quantization) in their quantization scheme instead of 1-bit quantization (binary quantization), then the $\alpha_{blind}$ increases from 0.5 to 0.75. This trend is observed with increasing number of quantization bits, and when the sensors employ an 8-bit quantizer, then the attacker needs to compromise at least 99.6\% of the sensors in the network to blind the FC. Obviously, the improvement in security performance is not free as the sensors incur a communication cost in terms of energy and bandwidth as the number of quantization bits increases. Therefore, in a practical world, the network designer faces a trade-off between the communication cost and the security guarantees.

Also, note that, when $M = 2$ (1-bit quantization), our results coincide with those of Rawat \emph{et al.} in \cite{Rawat2011}, where the focus was on the problem of binary hypotheses testing in a distributed sensor network. On the other hand, our results are more general as they address any inference problem - detection, estimation or classification in a distributed sensor network. Another extreme case to note is when $M \rightarrow \infty$, in which case, $\alpha_{blind} \rightarrow 1$. This means that the Byzantine attacker cannot blind the FC unless all the sensors are compromised.

In the following sections, we consider distributed detection and estimation problems in sensor networks and analyze the impact of the optimal Byzantine attack on these systems. For the sake of tractability, we consider a noiseless channel ($\mathbb{Q} = \mathbb{I}$) at the FC in the framework of resource-constrained Byzantine attack. Therefore, according to Theorem~\ref{thrm: Optimal Attack}, we restrict our attention to the set of highly-symmetric $\mathbb{P}$ for the sake of tractability. In other words, we assume that 
\begin{equation}
	p_{lm} = 
	\begin{cases}
		p & \text{if } l \neq m
		\\
		1 - (M-1)p & \text{otherwise.}
	\end{cases}
	\label{Eqn: Symmetric flipping matrix}
\end{equation}

\section{Distributed Detection in the Presence of Resource-Constrained Byzantine Attacks \label{sec: Detection}}

In this section, we consider a resource-constrained Byzantine attack on binary hypotheses testing in a distributed sensor network where the phenomenon of interest is denoted as $\theta$ and is modeled as follows:
\begin{equation}
	\theta = 
	\begin{cases}
		0; \mbox{ if } H_0
		\\
		1; \mbox{ if } H_1
	\end{cases}.
\end{equation}

In order to characterize the performance of the FC, we consider Kullback-Leibler Divergence (KLD) as the performance metric. Note that KLD can be interpreted as the error exponent in the Neyman-Pearson detection framework \cite{Chamberland2003}, which means that the probability of missed detection goes to zero exponentially with the number of sensors at a rate equal to KLD computed at the FC. We denote KLD at the FC by $\mathbb{D}_{FC}$ and define it as follows:
\begin{equation}
	\begin{array}{l}
		\mathbb{D}_{FC} \ = \ \displaystyle \mathbb{E}_{H_0} \left[ \log \left( \frac{P(\mathbf{v}|H_0)}{P(\mathbf{v}|H_1)} \right) \right] 
		\\
		\\
		\qquad \ \ = \displaystyle \sum_{\mathbf{m} \in \{1, \cdots, M\}^N} P(\mathbf{v} = \mathbf{m}|H_0) \cdot \log \left( \frac{P(\mathbf{v} = \mathbf{m}|H_0)}{P(\mathbf{v} = \mathbf{m}|H_1)} \right)
	\end{array}
	\label{Eqn: KLD defn.}
\end{equation}
Since we have assumed that the sensor observations are conditionally independent,\footnote{For notational convenience, sensor index $i$ is ignored in the rest of the paper.} KLD can be expressed as
\begin{equation}
	\mathbb{D}_{FC} = \displaystyle N D_{FC},
	\label{Eqn: KLD defn. - independent}
\end{equation}
where $$D_{FC} = \displaystyle \sum_{m = 1}^M P(v = m|H_0) \cdot \log \left( \frac{P(v=m|H_0)}{P(v = m|H_1)} \right).$$
Note that the optimal Byzantine attack, as given in Equation (\ref{Eqn: Optimal Attack}), results in equiprobable symbols at the FC irrespective of the hypotheses. Therefore, $\mathbb{D}_{FC} = 0$ under optimal Byzantine attack, resulting in the blinding of the FC.

On the other hand, if the attacker does not have enough resources to compromise $\alpha_{blind}$ fraction of sensors in the network (i.e. $\alpha < \alpha_{blind}$), an optimal strategy for the Byzantine node is to use an appropriate $\mathbb{P}$ matrix that deteriorates the performance of the sensor network to the maximal extent. As mentioned earlier in Section \ref{sec: Analysis}, we restrict our search to finding the optimal $\mathbb{P}$ within a space of highly symmetric row-stochastic matrices, as given in Equation \eqref{Eqn: Symmetric flipping matrix}. Thus, we formulate the problem as follows.

\begin{prob-statement}
	\label{Prob. Form: constrained attack - detection - symmetric}
	Given the value of $\alpha<\alpha_{blind}$, find the optimal $\mathbb{P}$ within a space of highly symmetric row-stochastic matrices, as given in Equation \eqref{Eqn: Symmetric flipping matrix}, such that
    \begin{equation*}
    	\begin{array}{ll}
    		\displaystyle \minimize_{p} & D_{FC}
        	\\ \text{subject to} & 0 \leq p \leq \displaystyle \frac{1}{M-1}
    	\end{array}
    \end{equation*}
\end{prob-statement}

Theorem \ref{Thrm: Optimal Byzantine Attack - constrained - symmetric} presents the optimal flipping probability that provides the solution to Problem \ref{Prob. Form: constrained attack - detection - symmetric}. Note that this result is independent of the design of the sensor network and, therefore, can be employed when the Byzantine has no knowledge about the network.

\begin{thrm}
	Given a fixed $\alpha < \displaystyle \frac{M-1}{M}$, the probability $p$ that optimizes $\mathbb{P}$ within a space of highly symmetric row-stochastic matrices, as given in Equation \eqref{Eqn: Symmetric flipping matrix}, such that $D_{FC}$ is minimized, is given by
	\begin{equation}
		p^* = \frac{1}{M-1}.
		\label{Eqn: Optimal flipping prob - symmetric}
	\end{equation}
	\label{Thrm: Optimal Byzantine Attack - constrained - symmetric}
\end{thrm}

\begin{proof}
See Appendix \ref{App: Theorem - Optimal Byzantine Attack - constrained - symmetric}.
\end{proof}

Note that this solution is of particular interest to the Byzantine attacker since the solution does not require any knowledge about the sensor network design. Also, the attacker's strategy is very simple to implement.

\subsection*{Numerical Results \label{sec: Numerical Results - Detection}}

For illustration purposes, let us consider the following example, where the inference network is deployed to aid the opportunistic spectrum access for a cognitive radio network (CRN). In other words, the CRs are sensing a licensed spectrum band to find the vacant band for the operation of the CRN.

Let the observation model at the $i^{th}$ sensor be defined as follows.
\begin{equation}
	r_i = s(\theta) + n_i,
	\label{Eqn: Observation Model - example - detection}
\end{equation}
where $\theta\in\{0,1\}$, $s(\theta) = \mu \cdot (-1)^{1 + \theta}$ is a BPSK-modulated symbol transmitted by the licensed (or the primary) user transmitter, and the noise $n_i$ is the AWGN at the $i^{th}$ sensor with probability distribution $\mathcal{N}(0, \sigma^2)$. 

Therefore, the conditional distribution of $r_i$ under $H_0$ and $H_1$ can be given as $\mathcal{N}(-\mu, \sigma^2)$ and $\mathcal{N}(\mu, \sigma^2)$ respectively. The range of $r_i$ spans the entire real line ($\mathbb{R}$). However, we assume that the quantizer restricts the support by limiting the range of output values to a smaller range, say $[-A, A]$. This parameter $A$ is called the \emph{overloading} parameter \cite{Book-Proakis-DSP} because the choice of $A$ dictates the amount of overloading distortion caused by the quantizer. Within this restricted range of observations, we assume a uniform quantizer with a step size (called the \emph{granularity} parameter) given by $\Delta = \frac{2}{M-2}$, which dictates the granularity distortion of the quantizer. In other words, the observation $r_i$ is quantized using the following quantizer:
\begin{equation}
	u_i = 
	\begin{cases}
		0; & \text{if } -\infty < r_i \leq \lambda_1
		\\
		1; & \text{if } \lambda_1 < r_i \leq \lambda_2
		\\
		\vdots
		\\
		M-1; & \text{if } \lambda_{M-1} < r_i \leq \infty
	\end{cases},
	\label{Eqn: Quantizer - Example - detection}
\end{equation}
where $$\lambda_i = A \cdot \left[ \frac{2(i-1)}{M-2} - 1 \right].$$ 

Note that, $\lambda_1 = -A$ and $\lambda_{M-1} = A$ represent the restricted range of the quantizer, as discussed earlier. The $i^{th}$ sensor transmits a symbol $v_i$ to the FC, where $v_i = u_i$ if it is honest. In the case of the $i^{th}$ sensor being a Byzantine node, the decision $u_i$ is modified into $v_i$ using the flipping probability matrix $\mathbb{P}$ as given in Equation (\ref{Eqn: Optimal Attack}).

\ignore{Note that we abstract the modulation of these quantized symbols in our paper. This is because of our assumption of an ideal channel between the sensors and the fusion center. In practice, each of these symbols are transformed into sequence of $\lceil \log_2(M) \rceil$ bits. These bits are later encoded using an appropriate error-correcting codebook, and modulated either using M-PSK/M-QAM constellation. The use of a powerful error-correcting code along with an appropriate modulation scheme ensures an almost-ideal channel at the FC.}

Although the performance of a given sensor network is quantified by the probability of error at the FC, we use a surrogate metric, as described earlier, called the KLD at the FC (Refer to Equation \eqref{Eqn: KLD defn.}) for the sake of tractability. In an asymptotic sense, Stein's Lemma \cite{Chamberland2003} states that the KLD is the rate at which the probability of missed detection converges to zero under a constrained probability of false alarm. Therefore, in our numerical results, we present how KLD at the FC varies with the fraction of Byzantine nodes $\alpha$, in the network.

\begin{figure}[!t]
	\centering
    \includegraphics[width=4in]{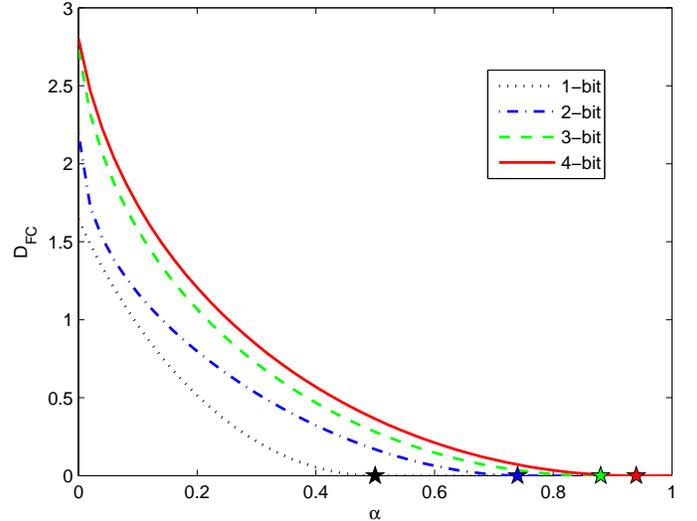}
    \caption{Contribution of a sensor to the overall KLD at the fusion center as a function of $\alpha$, for different number of quantization levels. The pentagrams on the x-axis correspond to the $\alpha_{blind}$ for 1-bit, 2-bit, 3-bit and 4-bit quantizations respectively from left to right.}
    \label{Fig: KLD_alpha_M}
\end{figure}

For the above sensor network, we assume that $\mu = 1$, $\sigma^2 = 1$ and $A = 2$. In Figure \ref{Fig: KLD_alpha_M}, we plot the contribution of each sensor in terms of KLD at the FC as a function of $\alpha$, for 1-bit, 2-bit, 3-bit and 4-bit quantizations, i.e., $M = $ 2, 4, 8 and 16 respectively, at the sensors. As per our intuition, we observe an improvement in both the detection performance (KLD) as well as security performance ($\alpha_{blind}$). Therefore, for a given $\alpha$, the Byzantine attack can be mitigated by employing finer quantization at the sensors. Of course, the best that the designer can do is to let the sensors transmit unquantized data to the FC, whether in the form of observation samples or their sufficient statistic (likelihood ratio). In this case, we can see that $\alpha_{blind} = 1$, since $\displaystyle \lim_{M \rightarrow \infty} \frac{M-1}{M} = 1$.

\section{Distributed Estimation in the Presence of Resource-Constrained Byzantine Attacks\label{sec: Estimation}}

In this section, we consider the problem of estimating a scalar parameter of interest, denoted by $\theta \in \mathbb{R}$, in a distributed sensor network. As described in the system model, we assume that the $i^{th}$ sensor quantizes its observation $r_i$ into an M-ary symbol $u_i$, and transmits $v_i$ to the FC. If the $i^{th}$ node is honest, then $v_i = u_i$. Otherwise, we assume that the sensor is compromised and flips $u_i$ into $v_i$ using a flipping probability matrix $\mathbb{P}$. Under the assumption that the FC receives the symbols $\mathbf{v}$ over an ideal channel, the estimation performance at the FC depends on the probability mass function $P(\mathbf{v}|\theta)$.

The performance of a distributed estimation network can be expressed in terms of the mean-squared error, defined as $\mathbb{E}\left[ (\hat{\theta} - \theta)^2 \right]$. In the case of unbiased estimators, this mean-squared error is lower bounded by the \emph{Cramer-Rao lower bound} (CRLB) \cite{Ribeiro2006}, which provides a benchmark for the design of an estimator at the FC. We present this result in Equation (\ref{Eqn: CRLB}):

\begin{equation}
	\displaystyle \mathbb{E}\left[ (\hat{\theta}(\mathbf{v}) - \theta)^2 \right] \geq \frac{1}{I_{FC}},
	\label{Eqn: CRLB}
\end{equation}
where
\begin{equation}
	\displaystyle I_{FC} = \mathbb{E}\left[ \left( \frac{\partial \log P(\mathbf{v}, \theta)}{\partial \theta} \right)^2 \right].
	\label{Eqn: FI}
\end{equation}
The term $I_{FC}$ is well known as the Fisher information (FI), and is, therefore, a performance metric that captures the performance of the optimal estimator at the FC. Note that, as shown in Equation \eqref{Eqn: FI Decomposition}, $I_{FC}$ can be further decomposed into two parts, one corresponding to the prior knowledge about $\theta$ at the FC, and the other (denoted as $J_{FC}$) representing the information about $\theta$, in the sensor transmissions $\mathbf{v}$:
\begin{equation}
	\displaystyle I_{FC} = J_{FC} + \mathbb{E}\left[ \left( \frac{\partial \log p(\theta)}{\partial \theta} \right)^2 \right],
	\label{Eqn: FI Decomposition}
\end{equation}
where
\begin{equation}
	\displaystyle J_{FC} = \mathbb{E}\left[ \left( \frac{\partial \log P(\mathbf{v} | \theta)}{\partial \theta} \right)^2 \right].
	\label{Eqn: FI decomposed}
\end{equation}

In most cases, a closed form expression for the mean-squared error is intractable and, therefore, conditional Fisher information (FI) is used as a surrogate metric to quantify the performance of a distributed estimation network. In this paper, we also use conditional FI of the received data $\mathbf{v}$ as the performance metric. Since the sensor observations are conditionally independent resulting in independent $\mathbf{v}$, we denote the conditional FI as $\mathbb{J}_{FC}$ and is defined as follows:
\begin{equation}
	\mathbb{J}_{FC} = N J_{FC},
	\label{Eqn: FI defn.}
\end{equation}
where 

\begin{equation}
	J_{FC} = \displaystyle \mathbb{E} \left[ \frac{\partial}{\partial \theta}\log P(\mathbf{v}|\theta) \right]^2 = - \mathbb{E} \left[ \frac{\partial^2}{\partial \theta^2}\log P(\mathbf{v}|\theta) \right].
	\label{Eqn: J_FC}
\end{equation}


Following the same approach as in Section \ref{sec: Detection}, we consider the problem of finding an optimal resource-constrained Byzantine attack when $\alpha < \alpha_{blind}$, by finding the symmetric transition matrix $\mathbb{P}$ that minimizes the conditional FI at the FC. This can be formulated as follows. 


%
%
%
\begin{prob-statement}
	\label{Prob. Form: constrained attack - estimation - symmetric}
	Given the value of $\alpha$, determine the optimal $\mathbb{P}$ within a space of highly symmetric row-stochastic matrices, as given in Equation \eqref{Eqn: Symmetric flipping matrix}, such that
    \begin{equation*}
    	\begin{array}{ll}
    		\displaystyle \minimize_{p} & J_{FC}
        	\\ \text{subject to} & 0 \leq p \leq \displaystyle \frac{1}{M-1}
    	\end{array}.
    \end{equation*}
\end{prob-statement}
Theorem \ref{Thrm: Optimal Byzantine Attack - estimation - constrained - symmetric} presents the optimal flipping probability that provides a solution to Problem \ref{Prob. Form: constrained attack - estimation - symmetric}. Note that this result is independent of the design of the sensor network and, therefore, can be employed when the Byzantine has no knowledge about the network.

\begin{thrm}
	Given a fixed $\alpha < \displaystyle \frac{M-1}{M}$, the flipping probability $p$ that optimizes $\mathbb{P}$ over a space of highly symmetric row-stochastic matrices, as given in Equation \eqref{Eqn: Symmetric flipping matrix}, by minimizing $J_{FC}$ is given by 
	\begin{equation*}
		p^* = \frac{1}{M-1}.
	\end{equation*}
	\label{Thrm: Optimal Byzantine Attack - estimation - constrained - symmetric}
\end{thrm}

\begin{proof}
See Appendix \ref{App: Theorem - Optimal Byzantine Attack - estimation - constrained - symmetric}.
\end{proof}

\subsection*{Numerical Results \label{sec: Numerical Results - Estimation}}

As an illustrative example, we consider the problem of estimating $\theta = 1$ at the FC based on all the sensors' transmitted messages. Let the observation model at the $i^{th}$ sensor be defined as follows:
\begin{equation}
	r_i = \theta + n_i,
	\label{Eqn: Observation Model - example - detection}
\end{equation}
where the noise $n_i$ is the AWGN at the $i^{th}$ sensor with probability distribution $\mathcal{N}(0, \sigma^2)$. The sensors employ the same quantizer as the one presented in Equation (\ref{Eqn: Quantizer - Example - detection}). The quantized symbol, denoted as $u_i$ at the $i^{th}$ sensor, is then modified into $v_i$ using the flipping probability matrix $\mathbb{P}$, as given in Equation (\ref{Eqn: Optimal Attack}).


Figure \ref{Fig: FI_alpha_M_middle}\ignore{ and \ref{Fig: FI_alpha_M_skewed}} plots the conditional FI corresponding to one sensor, for different values of $\alpha$ and $M$, when the uniform quantizer is centered around the true value of $\theta$. Note that as SNR increases ($\sigma \rightarrow 0$), we observe that it is better for the network to perform as much finer quantization as possible to mitigate the Byzantine attackers. On the other hand, if SNR is low, coarse quantization performs better for lower values of $\alpha$. This phenomenon of coarse quantization performing better under low SNR scenarios, can be attributed to the fact that more noise gets filtered as the quantization gets coarser (decreasing $M$) than the signal itself. On the other hand, in the case of high SNR, since the signal level is high, coarse quantization cancels out the signal component significantly, thereby resulting in a degradation in performance. 

\begin{figure*}[!t]
	\centerline
    {
    	\subfloat[Low SNR case: $\sigma = 1$]{\includegraphics[width=4in]{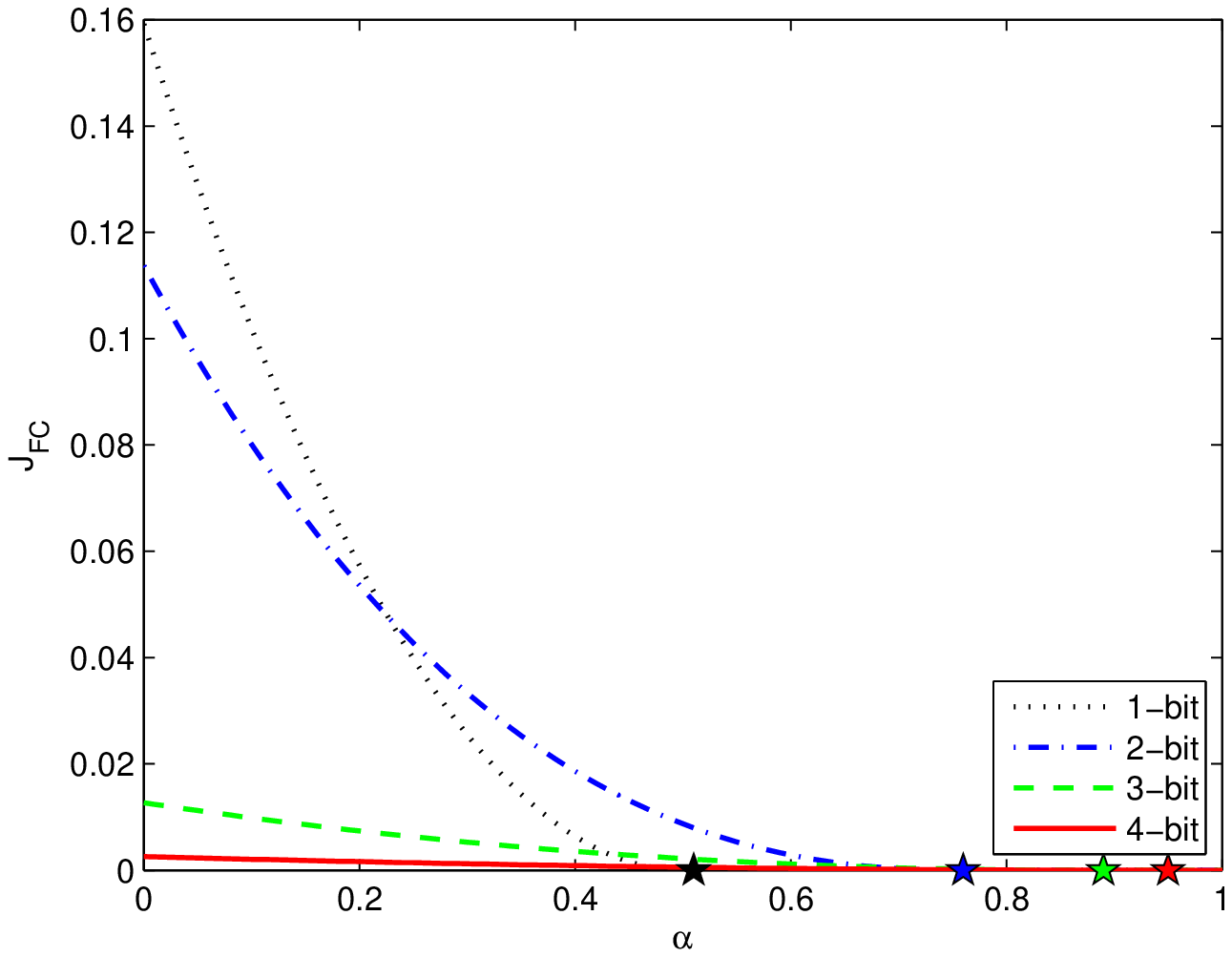}%
        \label{Fig: FI_alpha_M_lowSNR}}
        \hfil
        \subfloat[High SNR case: $\sigma = 0.01$]{\includegraphics[width=4in]{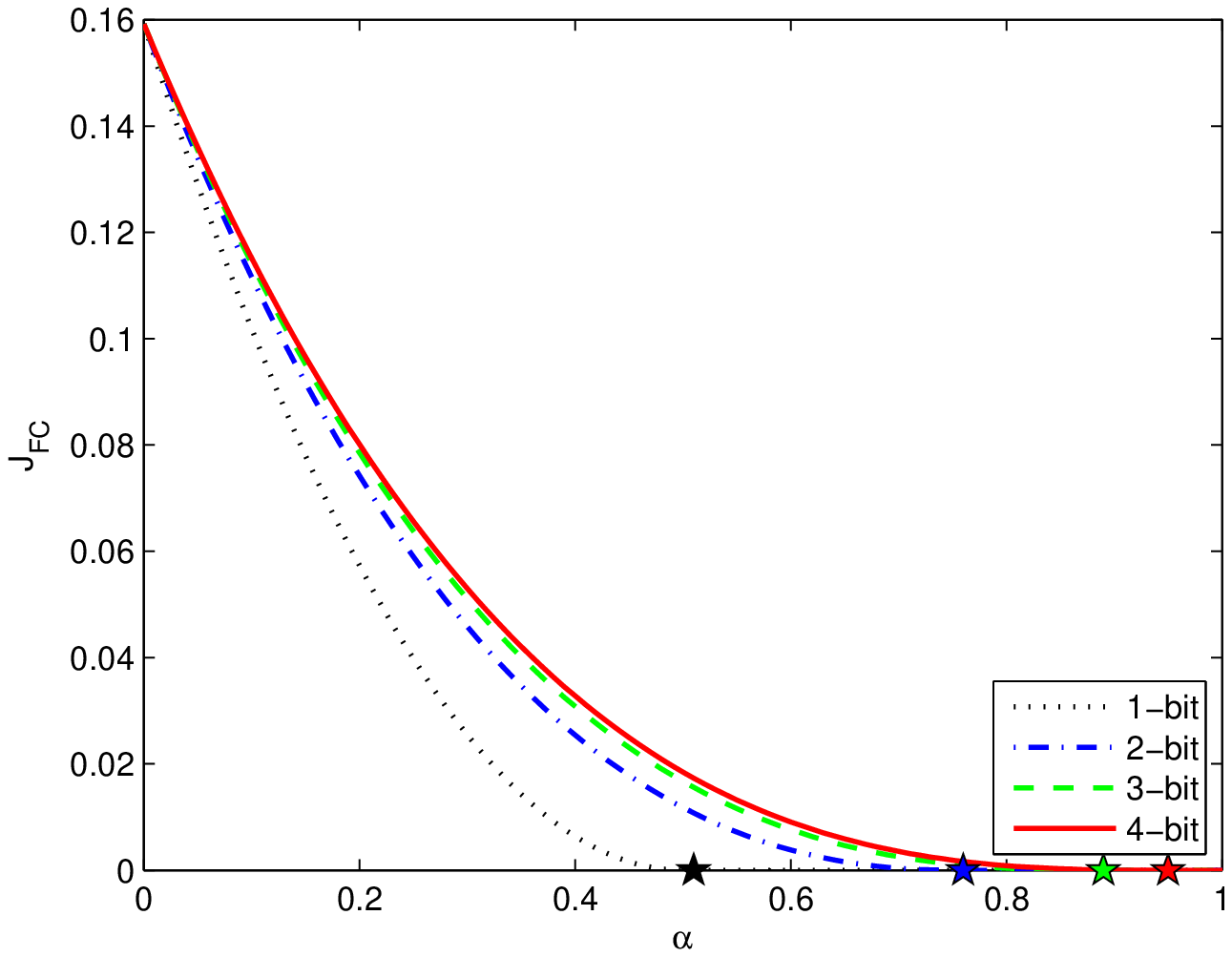}%
        \label{Fig: FI_alpha_M_highSNR}}
    }
    \caption{Contribution of a sensor to the overall conditional FI at the FC as a function of $\alpha$, for different number of quantization levels when $\theta = 0$ and $A = 2$. The pentagrams on the x-axis correspond to the $\alpha_{blind}$ for 1-bit, 2-bit, 3-bit and 4-bit quantizations respectively from left to right.}
    \label{Fig: FI_alpha_M_middle}
\end{figure*}


\section{Mitigation of Byzantine Attacks in a Bandwidth-Constrained Inference Network\label{sec: Reputation - Byzantine Identification}}

Given that the distributed inference network is under Byzantine attack, we showed that the performance of the network can be improved by increasing the quantization alphabet size of the sensors. Obviously, in a bandwidth-constrained distributed inference network, the sensors can only transmit with the maximum possible $M$, which is finite. In this section, we assume that the network cannot further increase the quantization alphabet size due to this bandwidth constraint. Therefore, we present a reputation-based Byzantine identification/mitigation scheme, which is an extension of the one proposed by Rawat \emph{et al.} in \cite{Rawat2011}, in order to improve the inference performance of the network. 

\subsection{Reputation-Tagging at the Sensors}
As proposed by Rawat \emph{et al.} in \cite{Rawat2011}, the FC identifies the Byzantine nodes by iteratively updating a reputation-tag for each node as time progresses. We extend the scheme to include fine quantization scenarios, i.e., $M > 2$, and analyze its performance through simulation results.

As mentioned earlier in the paper, the FC receives a vector $\mathbf{v}$ of received symbols from the sensors and fuses them to yield a global decision, denoted as $\hat{\theta}$. We assume that the observation model is known to the network designer, and is given as follows:
\begin{equation}
	r_i = f_i(\theta) + n_i,
	\label{Eqn: Model - Observation}
\end{equation}
where $f_i(\cdot)$ denotes the known observation model. We denote the quantization rule employed at the sensor as $\gamma$. Therefore, the quantized message at the sensor is given by $u_i = \gamma(r_i)$. As discussed earlier, the $i^{th}$ sensor flips $u_i$ into $v_i$ using a flipping probability matrix $\mathbb{P}$. 
Since the FC makes a global inference $\hat{\theta}$, it can calculate the squared-deviation $d_i$ of each sensor from the expected message that it is to nominally transmit as follows:
\begin{equation}
	d_i = \left( \gamma^{-1}(v_i) - f_i(\hat{\theta}) \right)^2,
	\label{Eqn: Reputation - noise estimate}
\end{equation}
where $\gamma^{-1}(v_i)$ is the inverse of the sensor quantizer $\gamma(v_i)$ and it  is assumed to be the centroid of the corresponding decision region of the quantizer $v_i$.

Note that $v_i$ is the received symbol which characterizes the behavior (honest or Byzantine) of the $i^{th}$ sensor, while $f_i(\hat{\theta})$ is the signal that the FC expects the sensor to observe. If the $i^{th}$ sensor is honest, we expect the mean of $d_i$ to be small. On the other hand, if the $i^{th}$ sensor is a compromised node, then the mean of $d_i$ is expected to be large. Therefore, we accumulate the squared-deviations $\mathbf{d_i} = \{ d_i(1), \cdots, d_i(T) \}$ over $T$ time intervals and compute a reputation tag $\Lambda_i(\mathbf{d_i})$, as a time-average for the $i^{th}$ node as follows:
\begin{equation}
	\Lambda_i = \displaystyle \frac{1}{T} \sum_{t = 1}^T d_i(t).
	\label{Eqn: Reputation - Tag}
\end{equation}  
The $i^{th}$ sensor is declared honest/Byzantine using the following threshold-based tagging rule
\begin{equation}
	\Lambda_i \quad \mathop{\stackrel{\text{Byzantine}}{\gtrless}}_{\text{Honest}} \quad \eta.
	\label{Eqn: Reputation - Threshold rule}
\end{equation}

The performance of the above tagging rule depends strongly on the choice of $\eta$. Note that the threshold $\eta$ should be chosen based on two factors. Firstly, $\eta$ should be chosen in such a way that the probability with which a malicious node is tagged Byzantine is high. Higher the value of $\eta$, lower is the chance of tagging a node to be Byzantine and vice-versa. This results in a tradeoff between the probability of detecting a Byzantine vs. the probability of falsely tagging an honest node as a Byzantine. Secondly, the value of $M$ also plays a role in the choice of $\eta$, and therefore, the performance of the tagging rule. We illustrate this phenomenon in our simulation results.

\begin{bluetext}
\subsection{Optimal Choice of the Tagging Threshold as $T \rightarrow \infty$}
In this paper, we denote the true type of the $i^{th}$ node as $\mathscr{T}_i$, where $\mathscr{T}_i = H$ corresponds to honest behavior, while $\mathscr{T}_i = B$ corresponds to Byzantine behavior, for all $i = 1, \cdots, N$. Earlier, in this section, we presented Equation (\ref{Eqn: Reputation - Threshold rule}) which allows us to make inferences about the true type. But, the performance of the Byzantine tagging scheme corresponding the $i^{th}$ sensor is quantified by the conditional probabilities $P(\Lambda_i \geq \eta | \mathscr{T}_i = \mathscr{T})$, for both $\mathscr{T} = H, B$. In order to find the optimal choice of $\eta$ in Equation (\ref{Eqn: Reputation - Threshold rule}), we continue with the Neyman-Pearson framework even in the context of Byzantine identification, where the goal is to maximize $P(\Lambda_i \geq \eta | \mathscr{T}_i = B)$, subject to the condition that $P(\Lambda_i \geq \eta | \mathscr{T}_i = H) \leq \xi$. 

To find these two conditional probabilities $P(\Lambda_i \geq \eta | \mathscr{T}_i = H)$ and $P(\Lambda_i \geq \eta | \mathscr{T}_i = B)$, we need a closed form expression of the conditional distributions, $P(\Lambda_i | \mathscr{T}_i = H)$ and $P(\Lambda_i | \mathscr{T}_i = B)$ respectively. In practice, where $T$ is finite, it is intractable to determine the conditional distribution of $\Lambda_i$, which is necessary to come up with the optimal choice of $\eta$. Therefore, in this paper, we assume that $T \rightarrow \infty$ and present an asymptotic choice of the tagging threshold $\eta$ used in Equation (\ref{Eqn: Reputation - Threshold rule}). 

As $T \rightarrow \infty$, since $d_i(t)$ is independent across $t = 1, \cdots, T$, due to central-limit theorem, $(\Lambda_i | \mathscr{T}_i = \mathscr{T}) \sim \mathcal{N}(\mu_{i,\mathscr{T}},\sigma_{i,\mathscr{T}})$, where 
\begin{equation}
	\begin{array}{lcl}
		\mu_{i,\mathscr{T}} & = & \displaystyle \mathbb{E}(\Lambda_i \ | \ \mathscr{T}_i = \mathscr{T})
		\\
		\\
		& = & \displaystyle \mathbb{E}\left[ \left( \gamma^{-1}(v_i(t)) - \hat{\theta}(t) \right)^2 | \ \mathscr{T}_i = \mathscr{T} \right]
	\end{array}
\end{equation}
and
\begin{equation}
	\begin{array}{lcl}
		\sigma_{i,\mathscr{T}}^2 & = & \displaystyle \var(\Lambda_i \ | \ \mathscr{T}_i = \mathscr{T})
		\\
		\\
		& = & \displaystyle \frac{1}{T} \var\left[ \left( \gamma^{-1}(v_i(t)) - \hat{\theta}(t) \right)^2 | \ \mathscr{T}_i = \mathscr{T} \right]
	\end{array}.
\end{equation}

In this paper, we do not present the final form of $\mu_{i,\mathscr{T}}$ and $\sigma_{i,\mathscr{T}}$ in order to preserve generality. Assuming that $v_i(t)$ is independent across sensors as well as time, the moments of $d_i$ can be computed for any given FC's inference $\hat{\theta}(t)$ at time $t$ about a given phenomenon. Although the final form of $\mu_{i,\mathscr{T}}$ and $\sigma_{i,\mathscr{T}}$ is not presented, since $d_i(t)$ is a function of $\mathbf{v}$, we present the conditional probability of $(v_j | \mathscr{T}_i = \mathscr{T})$ in Equation (\ref{Eqn: prob-v-node-type}), which is necessary for the computation of $\mu_{i,\mathscr{T}}$ and $\sigma_{i,\mathscr{T}}$. 
\begin{equation}
	\displaystyle P(v_j | \mathscr{T}_i = \mathscr{T}) = \int P(v_j | \theta, \mathscr{T}_i = \mathscr{T}) p(\theta) d\theta,
	\label{Eqn: prob-v-node-type}
\end{equation}
where $P(v_j | \theta, \mathscr{T}_i = \mathscr{T})$ can be calculated as follows:
\begin{equation}
	\begin{array}{l}
		\displaystyle P(v_j = m | \theta, \mathscr{T}_i = H) = 
		\\
		\\
		\qquad \qquad
		\begin{cases}
			\displaystyle P(u_j = m | \theta), & \mbox{if } j = i
			\\
			\\
			\displaystyle (1 - \pi_{BH}) P(u_j = m | \theta) \ +
			\\
			\quad \displaystyle \pi_{BH} \sum_{k = 1}^M p_{km} P(u_j = k | \theta), & \mbox{if } j \neq i
		\end{cases}
	\end{array}
	\label{Eqn: prob-v-node-type-honest-theta}
\end{equation}
and
\begin{equation}
	\begin{array}{l}
		\displaystyle P(v_j = m | \theta, \mathscr{T}_i = B) = 
		\\
		\\
		\qquad \qquad
		\begin{cases}
			\displaystyle \sum_{k = 1}^M p_{km} P(u_j = k | \theta), & \mbox{if } j = i
			\\
			\\
			\displaystyle (1 - \pi_{BB}) P(u_j = m | \theta) \ + 
			\\
			\quad \displaystyle \pi_{BB} \sum_{k = 1}^M p_{km} P(u_j = k | \theta), & \mbox{if } j \neq i
		\end{cases},
	\end{array}
	\label{Eqn: prob-v-node-type-Byzantine-theta}
\end{equation}
where $\pi_{BH} = P(\mathscr{T}_j = B | \mathscr{T}_i = H)$ and $\pi_{BB} = P(\mathscr{T}_j = B | \mathscr{T}_i = B)$ are conditional probabilities of the $j^{th}$ node's type, given the type of the $i^{th}$ node. Since there are $\alpha$ fraction of nodes in the network, given that the FC knows the type of $i^{th}$ node as $H$, the conditional probability of the $j^{th}$ node belonging to a type $\mathscr{T}$ is given by $\pi_{BH} = \displaystyle \frac{N\alpha}{N-1}$ and $\pi_{BB} = \displaystyle \frac{N\alpha - 1}{N-1}$.

Given the conditional distributions $P(\Lambda_i | \mathscr{T}_i = H)$ and $P(\Lambda_i | \mathscr{T}_i = B)$, we find the performance of the Byzantine identification scheme as follows:
\begin{equation}
	\begin{array}{lcl}
		P(\Lambda_i \geq \eta | \mathscr{T}_i = H) & = & \displaystyle Q \left( \frac{\eta - \mu_{i,H}}{\sigma_{i,H}} \right)
		\\
		\\
		P(\Lambda_i \geq \eta | \mathscr{T}_i = B) & = & \displaystyle Q \left( \frac{\eta - \mu_{i,B}}{\sigma_{i,B}} \right)
	\end{array}.
\end{equation}

Under the NP framework, the optimal $\eta$ can be chosen by letting $P(\Lambda_i \geq \eta | i = H) = \beta$, when $\Lambda_i$ is normally distributed conditioned on the true type of a given node. In other words,
\begin{equation}
	 \displaystyle Q \left( \frac{\eta - \mu_{i,H}}{\sigma_{i,H}} \right) = \xi
	\end{equation}
or equivalently,
\begin{equation}	
 \displaystyle \eta_{optimal} = \mu_{i,H} + \sigma_{i,H} Q^{-1}(\xi).
 \label{Eqn: Optimal Tagging Threshold}
\end{equation}

Note that, since $P(v_i | \mathscr{T}_i = H)$ is a function of $\alpha$, it follows that both $\mu_{i,H}$ and $\sigma_{i,H}$ are functions of $\alpha$. Although we do not provide a closed-form expression for $\eta$ as a function of $\alpha$, we provide the following example to portray how $\eta$ varies with different values of $\alpha$.

\subsubsection{Example: Variation of $\eta$ as a function of $\alpha$}
Consider a distributed estimation network with $N = 5$ identical nodes. Let the prior distribution of the true phenomenon $\theta$ be the uniform distribution $\mathcal{U}(0,1)$. We assume that the sensing channel is an AWGN channel where the sensor observations is given by $r_i = \theta + n_i$. Therefore, the conditional distribution of the sensor observations is $\mathcal{N}(\theta,\sigma^2)$, when conditioned on $\theta$. We assume that the sensors employ the quantizer rule shown in Equation (\ref{Eqn: Quantizer - Example - detection}) on their observations $r_i$. At the FC, we let $\gamma^{-1}(\cdot)$ be defined as the centroid function that returns $c_i = \displaystyle \frac{\lambda_{i-1} + \lambda_i}{2}$. Let $\hat{\theta} = \displaystyle \frac{1}{M} \sum_{i = 1}^N \gamma^{-1}(v_i(t))$ be the fusion rule employed at the FC to estimate $\theta$. 

Since the network comprises of identical nodes, without any loss of generality, we henceforth focus our attention on the reputation-based identification rule at sensor-1. Substituting the above mentioned fusion rule in the squared-deviation $d_1$ corresponding to sensor-1 in Equation (\ref{Eqn: Reputation - noise estimate}), we have
\begin{equation}
	\begin{array}{lcl}
		d_1 & = & \displaystyle \left( \gamma^{-1}(v_1) - \frac{1}{M} \sum_{i = 1}^5 \gamma^{-1}(v_i(t)) \right)^2
		\\
		\\
		& = & \displaystyle \left( \frac{M-1}{M} \gamma^{-1}(v_1) - \frac{1}{M} \sum_{i = 2}^5 \gamma^{-1}(v_i(t)) \right)^2.
	\end{array}
	\label{Eqn: Reputation - noise estimate - example}
\end{equation}
Let us denote $\phi_{ij} = \displaystyle \mathbb{E} \left\{ \left( \gamma^{-1}(v_i) \right)^j | \mathscr{T}_1 = H \right\} = \sum_{v_i = 1}^M \left[ \left( \gamma^{-1}(v_i) \right)^j P \left( v_i | \mathscr{T}_1 = H \right) \right]$, for all $i = 1, \cdots, 5$ and $j = 1, 2, \cdots, \infty$. Here, $ P \left( v_i | \mathscr{T}_1 = H \right)$ can be computed using Equation (\ref{Eqn: prob-v-node-type-honest-theta}) as follows:
\begin{equation}
	\begin{array}{l}
	 	P \left( v_i =m| \mathscr{T}_1 = H \right) 
	 	\\
	 	\\
	 	\qquad = \displaystyle \int_{-\infty}^\infty P \left( v_i=m | \theta, \mathscr{T}_1 = H \right) p(\theta) d\theta
	 	\\
	 	\\
	 	\qquad = \displaystyle \int_0^1 P \left( v_i=m | \theta, \mathscr{T}_1 = H \right) d\theta
		\\
		\\
		\ =
	 	\begin{cases}
	 		a_{1,m} & \mbox{if } i = 1
	 		\\
	 		\\
	 		\displaystyle \frac{N \alpha}{(N-1)(M-1)} + 
	 		\\
	 		\ \displaystyle \left(1 - \frac{M N \alpha}{(N-1)(M-1)} \right) a_{i,m} & \mbox{otherwise.}
	 	\end{cases}
 	\end{array}
\end{equation}
where $a_{i,m} = \displaystyle \int_0^1 P \left( u_i = m | \theta \right) d\theta$, for all $i = 1, \cdots, N$. Note that, since all the nodes in the network are identical, $P(u_i|\theta)$ is independent of the node-index $i$, and therefore, $\phi_{ij} = \phi_{2j}$, for all $i \neq 1$. 

Thus, the conditional mean and variance, $\mu_{1H}$ and $\sigma_{1H}^2$, are given as follows for the special case of $N=5$:
\begin{equation}
	\begin{array}{l}
		\mu_{1H} 
		\\
		\\
		\ = \displaystyle \mathbb{E}\left[ \left( \frac{M-1}{M} \gamma^{-1}(v_1) - \frac{1}{M} \sum_{i = 2}^5 \gamma^{-1}(v_i(t)) \right)^2 | \ \mathscr{T}_i = H \right]
		\\
		\\
		\ = \displaystyle \frac{1}{M^2} \mathbb{E}\left[ \left( (M-1) \gamma^{-1}(v_1) - \sum_{i = 2}^5 \gamma^{-1}(v_i(t)) \right)^2 | \ \mathscr{T}_i = H \right]
		\\
		\\
		\ = \displaystyle \frac{1}{M^2} \left[ (M-1)^2 \phi_{12} + 4 \phi_{22} + 12 \phi_{21}^2 - 8(M-1) \phi_{11} \phi_{21} \right]
	\end{array}
\end{equation}
and
\begin{equation}
	\begin{array}{lcl}
		\sigma_{1H}^2 & = & \displaystyle \frac{1}{T} \var\left[ \left( \gamma^{-1}(v_i(t)) - \hat{\theta}(t) \right)^2 | \ \mathscr{T}_i = H \right]
		\\
		\\
		& = & \displaystyle \frac{1}{T} \left\{ \Delta - \mu_{1H}^2 \right\},
	\end{array}
\end{equation}
where
\begin{equation}
	\begin{array}{l}
		\Delta 
		\\
		\\
		\ = \displaystyle \mathbb{E} \left[ \left( \frac{M-1}{M} \gamma^{-1}(v_1) - \frac{1}{M} \sum_{i = 2}^5 \gamma^{-1}(v_i(t)) \right)^4 | \ \mathscr{T}_i = H \right]
		\\
		\\
		\ = \displaystyle \frac{1}{M^4} \left[ (M-1)^4 \phi_{14} - 16(M-1)^3 \phi_{13} \phi_{21} \right.
		\\
		\qquad \qquad \displaystyle + 6(M-1)^2 \phi_{12} \{ 4 \phi_{22} + 12 \phi_{21}^2 \} 
		\\
		\qquad \qquad \displaystyle - 4(M-1)\phi_{11}(4\phi_{23} + 36 \phi_{22} \phi_{21} + 24 \phi_{21}^3) + 4 \phi_{24} 
		\\
		\qquad \qquad \displaystyle + 12 \phi_{23} \phi_{21} + 36(\phi_{23} \phi_{21} + \phi_{22}^2 + 2 \phi_{22} \phi_{21}^2) 
		\\
		\qquad \qquad \qquad \displaystyle \left. + 24(\phi_{21}^4 + 3 \phi_{22} \phi_{21}^2)\right].
	\end{array}
\end{equation}


\begin{figure*}[!t]
	\centerline
    {
    	\subfloat[$M = 2, \cdots, 7$]{\includegraphics[width=3.3in]{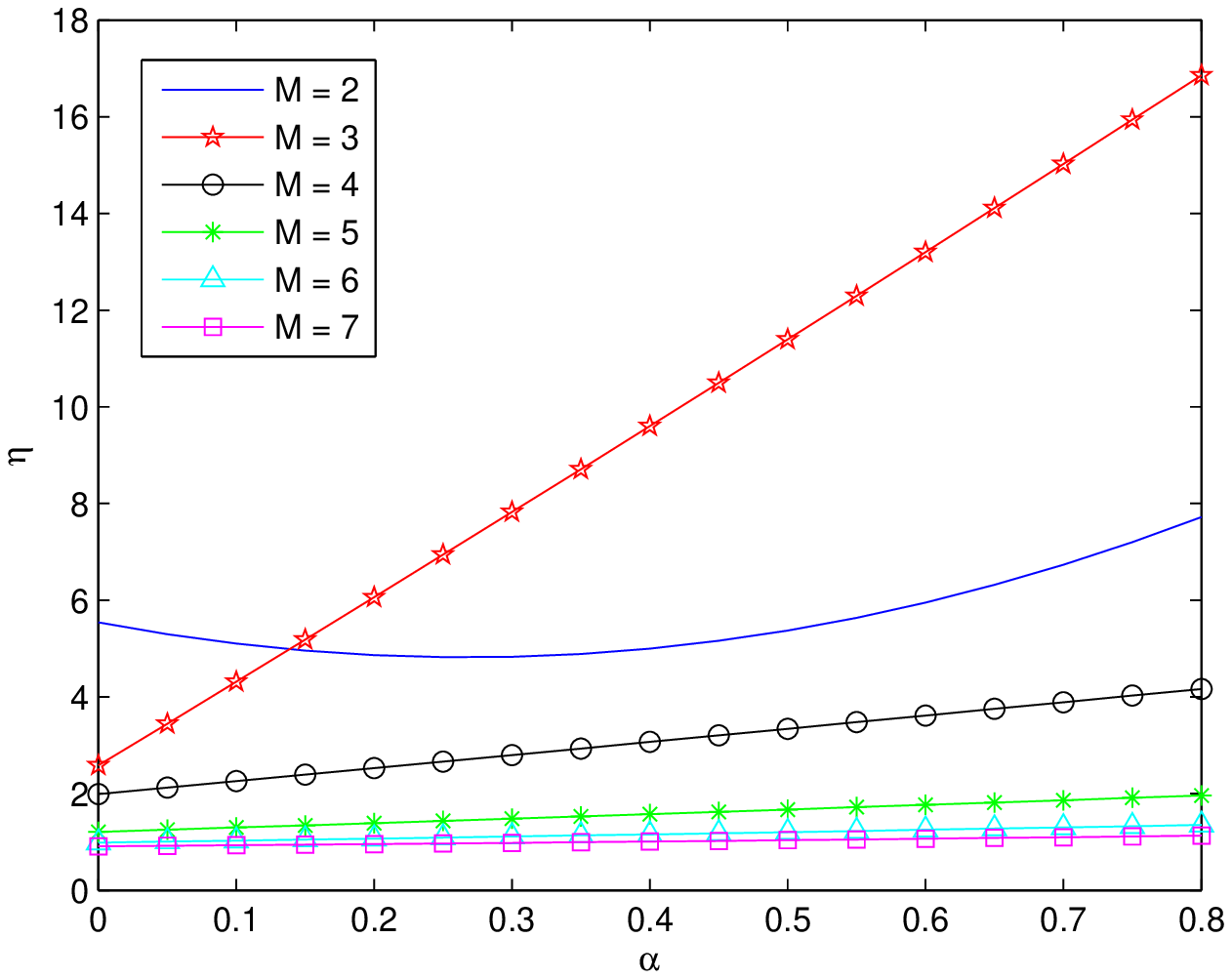}%
        \label{Fig: eta_vs_alpha_all}}
        \hfil
        \subfloat[$M = 7$]{\includegraphics[width=3.3in]{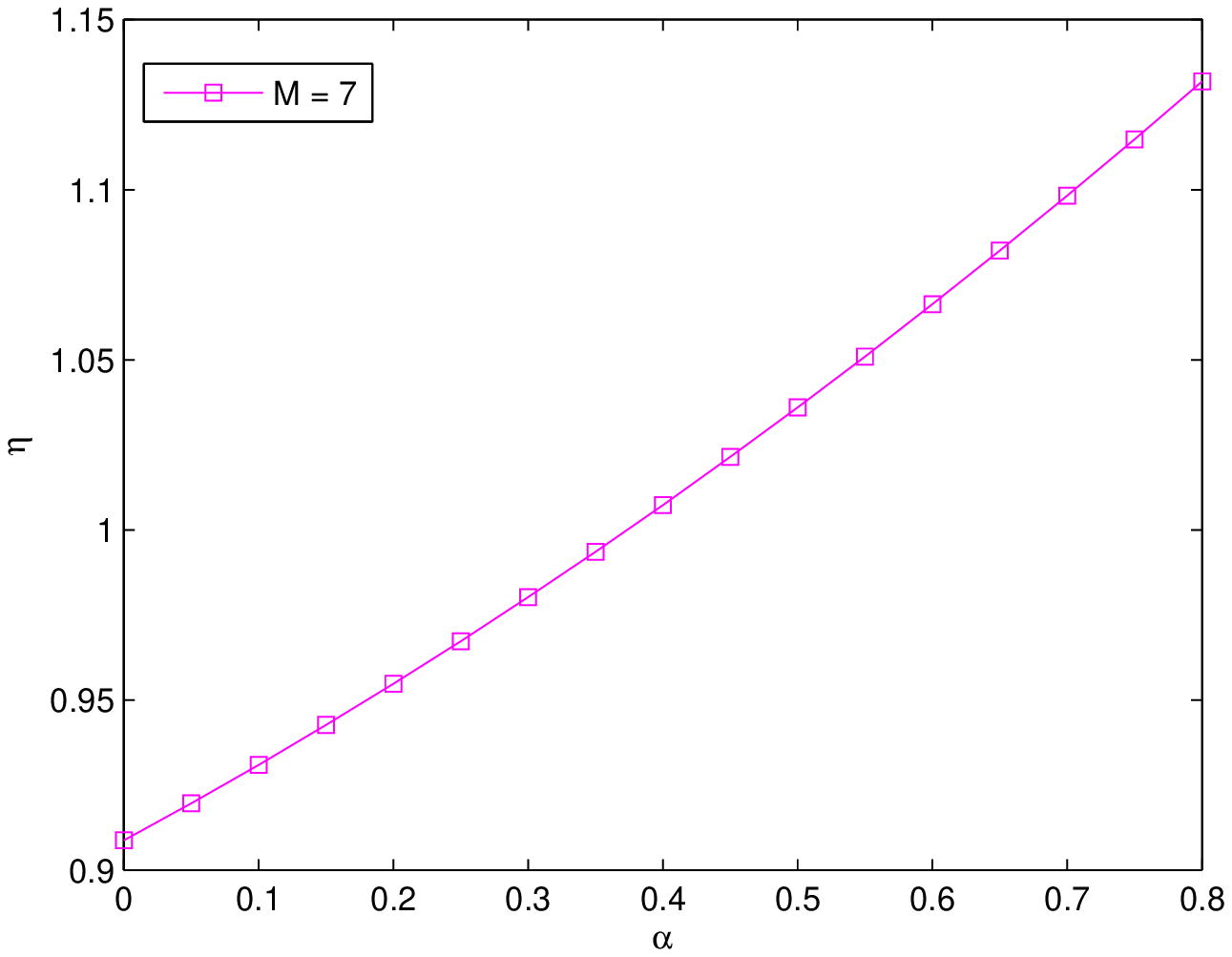}%
        \label{Fig: eta_vs_alpha_M_equals_7}}
    }
    \caption{Variation of the optimal tagging threshold $\eta$ (in the asymptotic sense, where $T \rightarrow \infty$) as a function of $\alpha$}
    \label{Fig: eta_vs_alpha}
\end{figure*}

Thus, for $\xi = 0.01$, we compute the tagging threshold $\eta$ numerically as shown in Equation (\ref{Eqn: Optimal Tagging Threshold}), and plot the variation of $\eta$ as a function of $\alpha$ in Figure \ref{Fig: eta_vs_alpha}. Note that, in our numerical results, we observe that the optimal choice of $\eta$ is a convex function of $\alpha$, where the curvature of the convexity decreases with increasing $M$. This can be clearly seen from Figure \ref{Fig: eta_vs_alpha_M_equals_7}, where we only plot the case of $M = 7$. We observe a similar behavior for all the other values of $M$, and therefore, present the case of $M = 7$ to illustrate the convex behavior of $\eta$. In other words, for very large values of $M$, the choice of $\eta$ becomes independent of $\alpha$, for any fixed $\alpha \leq \alpha_{blind}$.

\end{bluetext}

\subsection{Simulation Results}
In order to illustrate the performance of the proposed reputation-based scheme, we consider a sensor network with a total of 100 sensors in the network, out of which 20 are Byzantine sensors. Let the sensor quantizers be given by Equation (\ref{Eqn: Quantizer - Example - detection}) and the fusion rule at the FC be the MAP rule, given as follows:
\begin{equation}
	\displaystyle \sum_{i = 1}^N \log \left( \frac{P(v_i|H_1)}{P(v_i|H_0)} \right) \quad \mathop{\stackrel{\hat{\theta} = 1}{\gtrless}}_{\hat{\theta} = 0} \quad \log \frac{p_0}{p_1}.
	\label{Eqn: Fusion Rule}
\end{equation}

Figure \ref{Fig: Reputation - Identification} plots the rate of identification of the number of Byzantine nodes in the network for the proposed reputation-based scheme for different sizes of the quantization alphabet set. Note that the convergence rate deteriorates as $M$ increases. This is due to the fact that the Byzantine nodes have increasing number of symbol options to flip to, because of which a greater number of time-samples are needed to identify the malicious behavior. In addition, we also simulate the evolution of mislabelling an honest node as a Byzantine node in time, and plot the probability of the occurrence of this event in Figure \ref{Fig: Mislabel - Prob}. Just as the convergence deteriorates with increasing $M$, we observe a similar behavior in the evolution of the probability of mislabelling honest nodes. Another important observation in Figure \ref{Fig: Mislabel - Prob} is that the probability of mislabelling a node always converges to zero in time. Similarly, we simulate the evolution of mislabelling a Byzantine node as an honest one in time in Figure \ref{Fig: Byz - Mislabel - Prob}. We observe similar convergence of the probability of mislabelling a Byzantine node as an honest node to zero, with a rate that decreases with increasing number of quantization levels, $M$. Therefore, Figures \ref{Fig: Reputation - Identification}, \ref{Fig: Mislabel - Prob} and \ref{Fig: Byz - Mislabel - Prob} demonstrate that, after a sufficient amount of time, the reputation-based scheme always identifies the true behavior of a node within the network, with negligible number of mislabels. 

\begin{figure}[!t]
	\centering
    \includegraphics[width=3.5in]{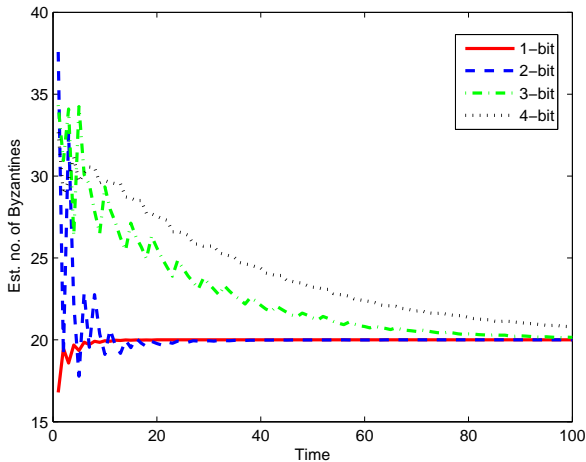}
    \caption{Rate of identification of the number of Byzantine nodes in time for different number of quantization levels}
    \label{Fig: Reputation - Identification}
\end{figure}

\begin{figure}[!t]
	\centering
    \includegraphics[width=3.5in]{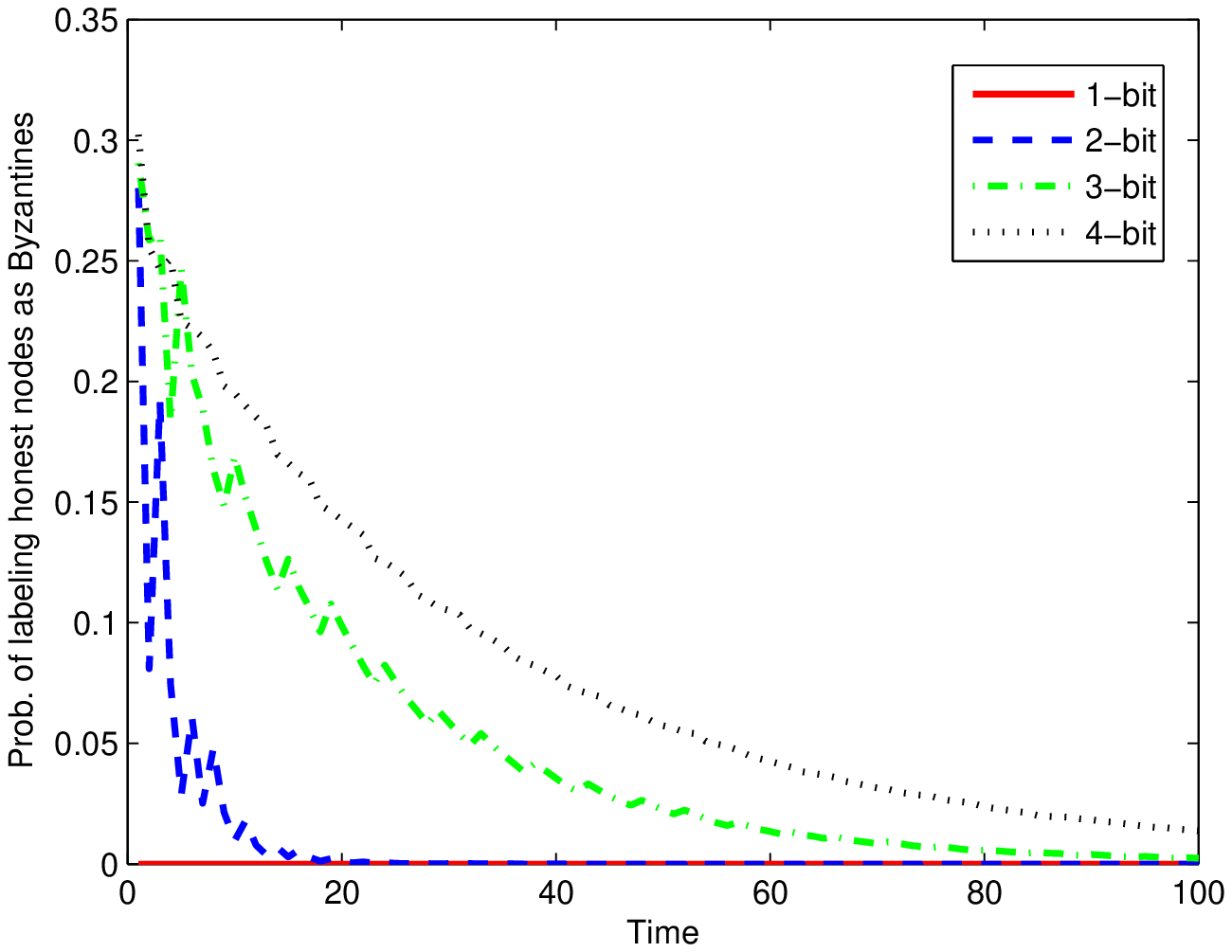}
    \caption{Evolution of the probability of mislabelling an honest node as a Byzantine in time for different number of quantization levels}
    \label{Fig: Mislabel - Prob}
\end{figure}

\begin{figure}[!t]
	\centering
    \includegraphics[width=3.5in]{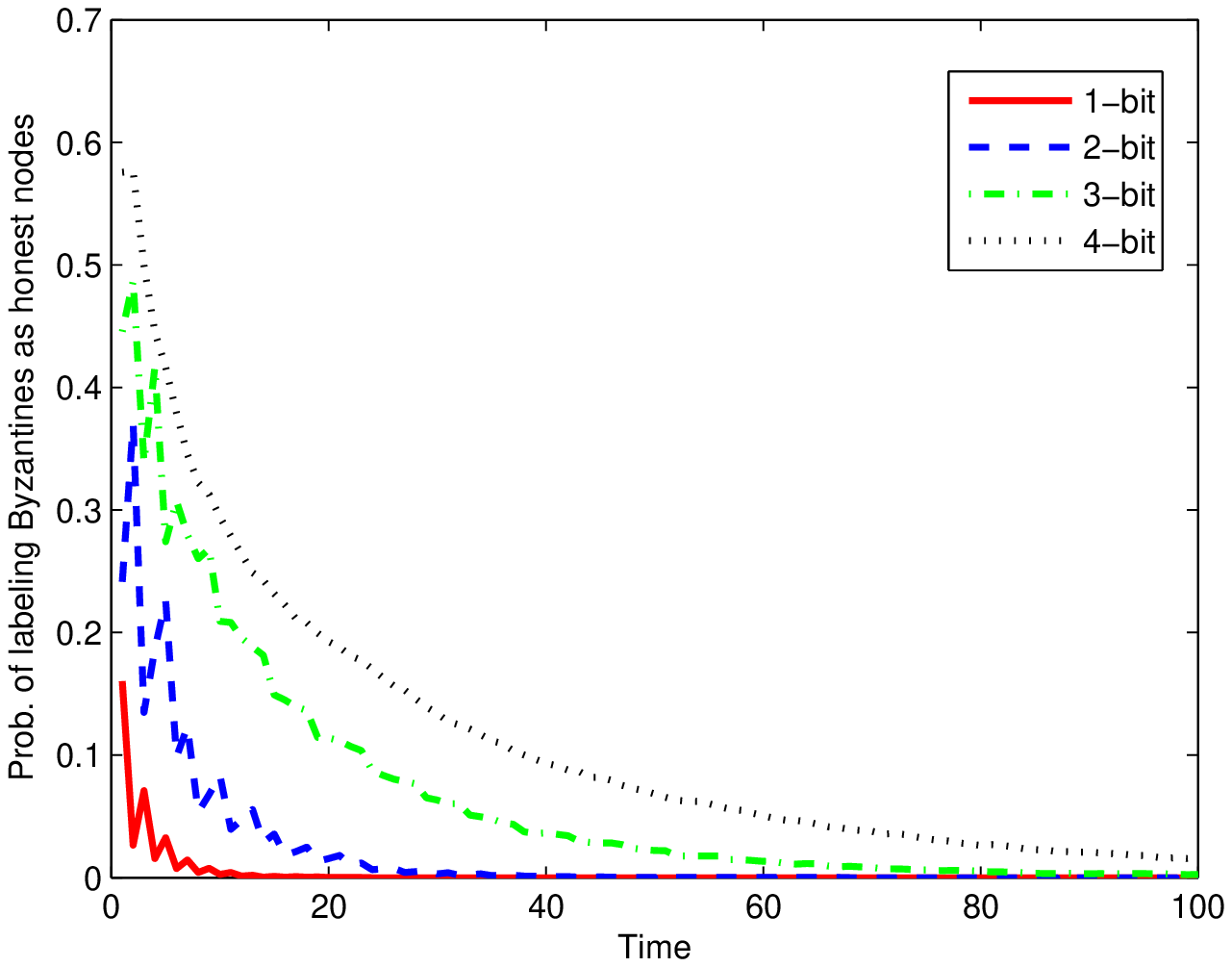}
    \caption{Evolution of the probability of mislabelling a Byzantine node as an honest node in time for different number of quantization levels}
    \label{Fig: Byz - Mislabel - Prob}
\end{figure}

\section{Concluding Remarks \label{sec: Conclusion}}
In summary, we modelled the problem of distributed inference with M-ary quantized data in the presence of Byzantine attacks, under the assumption that the attacker does not have knowledge about either the true hypotheses or the quantization thresholds at the sensors. We found the optimal Byzantine attack that \emph{blinds} the FC in the case of any inference task for both noiseless and noisy FC channels. We also considered the problem of resource-constrained Byzantine attack ($\alpha < \alpha_{blind}$) for distributed detection and estimation in the presence of resource-constrained Byzantine attacker for the special case of highly symmetric attack strategies in the presence of noiseless channels at the FC. From the inference network's perspective, we presented a mitigation scheme that identifies the Byzantine nodes through reputation-tagging. We also showed how the optimal tagging threshold can be found when the time-window $T \rightarrow \infty$. Finally, we also investigated the performance of our reputation-based scheme in our simulation results and show that our scheme always converges to finding all the compromised nodes, given sufficient amount of time. In our future work, we will investigate the optimal resource-constrained Byzantine attack in the space of all row-stochastic flipping probability matrices, and if possible, find schemes that mitigate the Byzantine attack more effectively. 

\appendices
\begin{bluetext}
\section{Optimal Byzantine Attack in the Presence of a Discrete Noisy Channel at the FC}
\label{sec:non-perfect channel}
Given that the messages $\mathbf{v} = \{ v_1, v_2, \cdots, v_N \}$ are transmitted to the fusion center (FC), we assume a discrete noise channel $\mathbb{Q} = [q_{mn}]$ between the sensors and the FC, where $q_{mn}$ is the probability with which $v_i = m$ is transformed to symbol $z_i = n$ at the $i^{th}$ sensor. Based on the received $\mathbf{z}$ at the FC, a global inference is made about the phenomenon of interest. In this paper, we assume that the row-stochastic channel matrix $\mathbb{Q}$ is invertible for the sake of tractability.

Given the transition probability matrix $\mathbb{Q}$ for the channel between the sensors and the FC, we assume that the FC receives $z_i = n$ when the the $i^{th}$ sensor transmits $v_i = m$, with a probability $q_{mn}$. The conditional distribution of $z_i = n$ under a given phenomenon $\theta$, is given as 
\begin{equation}
	P(z_i = n|\theta) = \displaystyle \sum_{m=1}^M q_{mn} P(v_i = m|\theta).
	\label{Eqn: Cond. Prob zi given theta}
\end{equation}
Note that if $\mathbb{Q}$ is a doubly stochastic matrix, since $\displaystyle \sum_{m=1}^M q_{mn} = 1$, it is sufficient for the Byzantine attacker to ensure $P(v_i = m|\theta) = \frac{1}{M}$. Thus, by Theorem~\ref{thrm: Optimal Attack}, we have the following theorem when $\mathbb{Q}$ is a doubly stochastic matrix.

\begin{thrm}
	If the channel matrix $\mathbb{Q}$ is doubly-stochastic, and if the Byzantine attacker has no knowledge about the sensors' quantization thresholds, then the optimal Byzantine attack is given as 
\begin{equation}
	\begin{array}{lcl}
		p_{lm} & = & \begin{cases} \displaystyle \frac{1}{M-1} & ; \text{ if } l \neq m \\ 0 &; \text{ otherwise} \end{cases}
		\\
		\\
		\alpha_{blind} & = & \displaystyle \frac{M-1}{M}.
		\\
	\end{array}
	\label{Eqn: Optimal Attack - Doubly Stochastic Channel}
\end{equation}
\label{Thrm: Optimal Attack - Doubly Stochastic Channel}
\end{thrm}

Therefore, we focus our attention to any general row-stochastic channel matrix $\mathbb{Q}$, where $\displaystyle \sum_{m=1}^M q_{mn}$ need not necessarily sum to unity for all $n = 1, \cdots, M$. In other words, the Byzantine attacker has to find an alternative strategy to blind the FC, where $P(z_i = n | \theta) = \frac{1}{M}$. Substituting Equation (\ref{Eqn: Cond. Prob vi given theta}) in Equation (\ref{Eqn: Cond. Prob zi given theta}) and rearranging the terms, we have the following.
\begin{equation}
	\begin{array}{l}
		P(z_i = n|\theta) = \displaystyle \sum_{m=1}^M q_{mn} P(v_i = m|\theta)
		\\
		\\
		\ = \displaystyle \sum_{m=1}^M q_{mn} [(1 - \alpha) + \alpha p_{mm}]
		\\
		\ + \displaystyle \sum_{m=1}^M q_{mn} \left\{ \sum_{l \neq m}\{\alpha p_{lm} -[(1 - \alpha) + \alpha p_{mm}]\} P(u_i = l|\theta) \right\}
		\\
		\\
		\ = \displaystyle \sum_{m=1}^M q_{mn} [(1 - \alpha) + \alpha p_{mm}]
		\\
		\ + \displaystyle \sum_{l=1}^M \left[ \sum_{m \neq l} q_{mn} \{\alpha p_{lm} -[(1 - \alpha) + \alpha p_{mm}]\} \right] P(u_i = l|\theta).
	\end{array}
	\label{Eqn: Cond. Prob zi given theta - expanded}
\end{equation}

The goal of a Byzantine attack is to blind the FC with the least amount of effort (minimum $\alpha$). To totally blind the FC is equivalent to making $P(z_i = n|\theta) = 1/M$ for all $0 \leq n \leq M-1$. In Equation (\ref{Eqn: Cond. Prob zi given theta - expanded}), the RHS consists of two terms. The first one is based on prior knowledge and the second term conveys information based on the observations. In order to blind the FC, the attacker should make the second term equal to zero. Since the attacker does not have any knowledge regarding $P(u_i = l|\theta)$, it can make the second term of Equation (\ref{Eqn: Cond. Prob zi given theta - expanded}) equal to zero by setting
\begin{equation}
	\displaystyle \sum_{m \neq l} q_{mn} \{\alpha p_{lm} -[(1 - \alpha) + \alpha p_{mm}]\} = 0 \mbox{ for all } 1 \leq n, l \leq M.
	\label{Eqn: Blinding condition - Noisy Channel}
\end{equation}
Then the conditional probability $\displaystyle P(z_i = n|\theta) = \displaystyle \sum_{m=1}^M q_{mn} [(1 - \alpha) + \alpha p_{mm}]$ becomes independent of the observations $r_i$ (or its quantized version $u_i$), resulting in equiprobable symbols at the FC. In other words, the received vector $\mathbf{z} = \{ z_1, z_2, \cdots, z_N \}$ does not carry any information about $\mathbf{u} = \{ u_1, u_2, \cdots, u_N \}$, thus making FC solely dependent on its prior information about $\theta$ in making an inference.

In order to identify the strategy that the attacker should employ to achieve the condition in Equation (\ref{Eqn: Blinding condition - Noisy Channel}), for all $n = 1, \cdots, M$, we need
\begin{equation}
	\begin{array}{lrcl}
		& P(z_i = n|\theta) & = & \displaystyle \frac{1}{M},
		\\
		\\
		\mbox{or,} & \displaystyle \sum_{m=1}^M q_{mn} \left\{  (1 - \alpha) + \alpha p_{mm}\right\} & = & \displaystyle \frac{1}{M}.
	\end{array}
	\label{Eqn: alpha-blind-condition-noisy-channel}
\end{equation}
In matrix form, we can rewrite Equation (\ref{Eqn: alpha-blind-condition-noisy-channel}) as 
\begin{equation*}
	\displaystyle (1 - \alpha) \mathbf{1}^T \mathbb{Q}  + \alpha \mathbf{p}^T \mathbb{Q} = \displaystyle \frac{1}{M} \mathbf{1}^T,
\end{equation*}
where $\mathbf{1}$ is an all-one column-vector and $\mathbf{p} = \left[ p_{11}, \cdots, p_{MM} \right]^T$ is the column-vector of all diagonal elements of $\mathbb{P}$.
In other words,
\begin{equation}
	\displaystyle \alpha (\mathbf{1} - \mathbf{p}) = \displaystyle \mathbf{1} - \frac{1}{M} \left(\mathbb{Q}^T \right)^{-1} \mathbf{1} 
	\label{Eqn: alpha-blind-condition-noisy-channel-2}
\end{equation}

Note that every element in the LHS of Equation (\ref{Eqn: alpha-blind-condition-noisy-channel-2}) always lies between 0 and 1. Therefore, the existence of the Byzantine's optimal strategy relies on the following condition.
In other words,
\begin{equation}
	\displaystyle \mathbf{0} \quad \leq \quad \left(\mathbb{Q}^T \right)^{-1} \mathbf{1} \quad \leq \quad M \ \mathbf{1}.
	\label{Eqn: alpha-blind-condition-noisy-channel-existence}
\end{equation}If \eqref{Eqn: alpha-blind-condition-noisy-channel-existence} does not hold, there does not exist an optimal strategy. 
Given that the condition in Equation (\ref{Eqn: alpha-blind-condition-noisy-channel-existence}) holds, the minimum $\alpha$ can be found as follows.
\begin{equation}
	\begin{array}{lcl}
		\displaystyle \alpha_{blind} & = & \displaystyle \min \left\{ \mathbf{1} - \frac{1}{M} \left(\mathbb{Q}^T \right)^{-1} \mathbf{1} \right\}
		\\
		\\
		& = & \displaystyle 1 - \frac{1}{M} \max \left\{ \left(\mathbb{Q}^T \right)^{-1} \mathbf{1} \right\}.
	\end{array}
	\label{Eqn: alpha-blind-condition-noisy-channel-3}
\end{equation}
Therefore, $\mathbf{p}$ can be calculated as
\begin{equation}
	 \begin{array}{lcl}
	 	\mathbf{p} & = & \displaystyle \mathbf{1} - \frac{1}{\alpha_{blind}}\left( \mathbf{1} -  \frac{1}{M} \left(\mathbb{Q}^T \right)^{-1} \mathbf{1} \right)
	 	\\
	 	\\
	 	& = & \displaystyle \frac{1}{\alpha_{blind}M}\left(\mathbb{Q}^T \right)^{-1} \mathbf{1}-\frac{1-\alpha_{blind}}{\alpha_{blind}}\mathbf{1}.
	 \end{array}
	\label{Eqn: diagonal-P-noisy-channel}
\end{equation}

Next, in order to find the rest of the $\mathbb{P}$ matrix, let us consider Equation (\ref{Eqn: Blinding condition - Noisy Channel}). Adding $q_{ln} \left\{ \alpha p_{ll} - [1 - \alpha + \alpha p_{ll}] \right\}$ on both sides to Equation (\ref{Eqn: Blinding condition - Noisy Channel}), we have
\begin{equation}
	\begin{array}{ll}
		& \displaystyle \sum_{m = 1}^M q_{mn} \{\alpha p_{lm} -[(1 - \alpha) + \alpha p_{mm}]\}  = - q_{ln} (1-\alpha) 
		\\		
		& \qquad \qquad \qquad \qquad \qquad \qquad \mbox{ for all } 1 \leq n, l \leq M.
		\\
		\\
		\mbox{or,} & \displaystyle \alpha \sum_{m = 1}^M q_{mn} p_{lm} = \displaystyle \frac{1}{M}- q_{ln} (1-\alpha)
		\\
		& \qquad \qquad \qquad \qquad \qquad \qquad \mbox{ for all } 1 \leq n, l \leq M.
	\end{array}
	\label{Eqn: Blinding condition - Noisy Channel - 2}
\end{equation}
In matrix form, we have
\begin{equation}
	\displaystyle \alpha \mathbb{P}\mathbb{Q} = \frac{1}{M}\mathds{1} -(1-\alpha)  \mathbb{Q},
	\label{Eqn: Blinding condition - Noisy Channel - 3}
\end{equation}
where $\mathds{1}$ is an all-one matrix.
Equivalently, we have
\begin{equation}
		 \mathbb{P}  =  \displaystyle \frac{1}{\alpha M}\mathds{1}\mathbb{Q}^{-1} - \frac{1-\alpha}{\alpha} \mathbb{I},
	\label{Eqn: Blinding condition - Noisy Channel - 4}
\end{equation}
where $\mathbb{I}$ is the identity matrix. Note that the vector $\mathbf{p}$ (comprising the diagonal elements of $\mathbb{P}$) obtained from Equation \eqref{Eqn: Blinding condition - Noisy Channel - 4} is verified to be same as that from Equation \eqref{Eqn: diagonal-P-noisy-channel}.

In summary, we have the following theorem that provides the optimal Byzantine strategy in the presence of noisy FC channels:
\begin{thrm}
	Let the Byzantine attacker have no knowledge about the sensors' quantization thresholds, and, the FC's channel matrix be $\mathbb{Q}$. If $\mathbb{Q}$ is non-singular, and, if $\displaystyle \ \mathbf{0} \ \leq \ \left(\mathbb{Q}^T \right)^{-1} \mathbf{1} \ \leq \ M \mathbf{1}$, then the optimal Byzantine attack is given as 
\begin{equation}
	\begin{array}{lcl}
		\alpha_{blind} & = &  \displaystyle 1 - \frac{1}{M} \max \left\{ \left(\mathbb{Q}^T \right)^{-1} \mathbf{1} \right\}
		\\
		\\
		\mathbb{P}  &=&  \displaystyle \frac{1}{\alpha_{blind} M}\mathds{1}\mathbb{Q}^{-1} - \displaystyle\frac{1-\alpha_{blind}}{\alpha_{blind}} \mathbb{I}.
	\end{array}
	\label{Eqn: Optimal Attack - Row Stochastic Channel}
\end{equation}
\label{Thrm: Optimal Attack - Row Stochastic Channel}
\end{thrm}
Note that, if the channel matrix $\mathbb{Q}$ is doubly-stochastic, we have $\mathbb{Q} \mathbf{1} = \mathbf{1}$ and $\mathbb{Q}^T \mathbf{1} = \mathbf{1}$. Substituting these conditions in Equation (\ref{Eqn: Optimal Attack - Row Stochastic Channel}), Theorem \ref{Thrm: Optimal Attack - Row Stochastic Channel} reduces to Theorem \ref{Thrm: Optimal Attack - Doubly Stochastic Channel}.

Having identified the optimal Byzantine attack, one can observe that the attacker needs to compromise a huge number of sensors ($\alpha_{blind} = \displaystyle 1 - \frac{1}{M} \max \left\{ \left(\mathbb{Q}^T \right)^{-1} \mathbf{1} \right\}$) in the network to blind the FC. Therefore, it is obvious that, in the case of a resource-constrained attacker, the attacker compromises a fixed fraction of nodes $\alpha \leq \alpha_{blind}$ in such a way that the performance degradation at the FC is maximized. In our future work, we will investigate the problem of finding the optimal strategy in the context of resource-constrained Byzantine attacks in the presence of noisy FC channels.


%
%
%

\end{bluetext}

\section{Proof for Theorem \ref{Thrm: Optimal Byzantine Attack - constrained - symmetric} \label{App: Theorem - Optimal Byzantine Attack - constrained - symmetric}}

%
For the sake of notational simplicity, let us denote $x_m = P(u = m|H_0)$ and $y_m = P(u = m|H_1)$. Similarly, $\tilde{x}_m = P(v = m|H_0)$ and $\tilde{y}_m = P(v = m|H_1)$.

Rewriting Equation~\eqref{Eqn: Cond. Prob vi given theta} in our new notation, we have
\begin{equation}
	\tilde{x}_m = \alpha\sum_{l\neq m}p x_l + (1-\alpha(M-1)p) x_m = \alpha p + (1-M\alpha p)x_m
	\label{head_x_m}
\end{equation}
and
\begin{equation}
	\tilde{y}_m = \alpha\sum_{l\neq m}p y_l + (1-\alpha(M-1)p) y_m = \alpha p +(1-M\alpha p)y_m.
	\label{head_y_m}
\end{equation}

Therefore, the KLD at the FC can be rewritten as
\begin{equation}
	\displaystyle D_{FC} = \displaystyle \sum_{m = 1}^M \tilde{x}_m \log \left( \frac{\tilde{x}_m}{\tilde{y}_m} \right).
	\label{Eqn: D_FC - x - y}
\end{equation}
 

On partially differentiating $D_{FC}$ with respect to $p$, we have

\begin{equation}
	\begin{array}{l}
		\displaystyle \frac{\partial D_{FC}}{\partial p} \ = \ \displaystyle \frac{\partial}{\partial p} \sum_{m = 1}^M \tilde{x}_m \log \left( \frac{\tilde{x}_m}{\tilde{y}_m} \right)
		\\
		\\
		\quad = \displaystyle \alpha \sum_{m = 1}^M \left[ (1 - M x_m) \left( 1 + \log \frac{\tilde{x}_m}{\tilde{y}_m} \right) - (1 - M y_m) \frac{\tilde{x}_m}{\tilde{y}_m} \right]
		\\
		\\
		\quad = \displaystyle \alpha \sum_{m = 1}^M (1 - M x_m) + \alpha \sum_{m = 1}^M (1 - M x_m) \log \frac{\tilde{x}_m}{\tilde{y}_m} 
		\\ \\ \qquad \qquad \qquad - \displaystyle \alpha \sum_{m = 1}^M (1 - M y_m) \frac{\tilde{x}_m}{\tilde{y}_m}.
	\end{array}
	\label{Eqn: Diff - D_FC - p - ineq}
\end{equation}

Consider the first term in the RHS of Equation (\ref{Eqn: Diff - D_FC - p - ineq}). Note that, since $\mathbf{x} = \{ x_1, \cdots, x_M \}$ is a probability mass function, we have 
$$\displaystyle \sum_{m = 1}^M (1 - M x_m) = M - M \sum_{m = 1}^M x_m = M - M = 0.$$
Therefore, Equation (\ref{Eqn: Diff - D_FC - p - ineq}) reduces to
\begin{equation}
	\displaystyle \frac{\partial D_{FC}}{\partial p} = \displaystyle \alpha \sum_{m = 1}^M (1 - M x_m) \log \frac{\tilde{x}_m}{\tilde{y}_m} - \alpha \sum_{m = 1}^M (1 - M y_m) \frac{\tilde{x}_m}{\tilde{y}_m}.
	\label{Eqn: Diff - D_FC - p - ineq - 2}
\end{equation}

Rearranging the terms in Equation (\ref{Eqn: Diff - D_FC - p - ineq - 2}), we have
\begin{equation}
	\begin{array}{l}
		\displaystyle \frac{\partial D_{FC}}{\partial p} \ = \ \displaystyle \alpha \sum_{m = 1}^M \left[ \log \frac{\tilde{x}_m}{\tilde{y}_m} - \frac{\tilde{x}_m}{\tilde{y}_m} \right] - \alpha M \sum_{m = 1}^M x_m \log \frac{\tilde{x}_m}{\tilde{y}_m} 
		\\ \\ \qquad \qquad \qquad \qquad \displaystyle + \alpha M \sum_{m = 1}^M y_m \frac{\tilde{x}_m}{\tilde{y}_m}.
	\end{array}
	\label{Eqn: Diff - D_FC - p - ineq - 3}
\end{equation}

Let us denote the first term as $T_1$. In other words, 

$$T_1 = \alpha \sum_{m = 1}^M \left[ \log \frac{\tilde{x}_m}{\tilde{y}_m} - \frac{\tilde{x}_m}{\tilde{y}_m} \right].$$

Let us now focus our attention on the other terms in the RHS of Equation (\ref{Eqn: Diff - D_FC - p - ineq - 3}). Substituting Equations \eqref{head_x_m} and \eqref{head_y_m} in the second and third terms of the RHS of Equation (\ref{Eqn: Diff - D_FC - p - ineq - 3}), we have 
\begin{equation}
	\begin{array}{l}
		\displaystyle \frac{\partial D_{FC}}{\partial p} \ = \ \displaystyle T_1 - \frac{M \alpha}{1 - M \alpha p} \sum_{m = 1}^M (\tilde{x}_m - \alpha p) \log \frac{\tilde{x}_m}{\tilde{y}_m} 
		\\ \\ \qquad \qquad \qquad \qquad \displaystyle + \frac{M \alpha}{1 - M \alpha p} \sum_{m = 1}^M (\tilde{y}_m - \alpha p) \frac{\tilde{x}_m}{\tilde{y}_m}
		\\
		\\
		\quad = \ \displaystyle T_1 - \frac{M \alpha}{1 - M \alpha p} D(\mathbf{\tilde{x}} || \mathbf{\tilde{y}}) 
		\\ \\ \qquad \displaystyle + \frac{M \alpha}{1 - M \alpha p} \left\{ \sum_{m = 1}^M \alpha p \log \frac{\tilde{x}_m}{\tilde{y}_m} - \sum_{m = 1}^M \alpha p \frac{\tilde{x}_m}{\tilde{y}_m} + \sum_{m = 1}^M \tilde{x}_m \right\},
	\end{array}
	\label{Eqn: Diff - D_FC - p - ineq - 4}
\end{equation}
where $D(\mathbf{\tilde{x}}||\mathbf{\tilde{y}})$ is the KLD between $\mathbf{\tilde{x}}$ and $\mathbf{\tilde{y}}$ and is, therefore, non-negative. Also, note that in Equation (\ref{Eqn: Diff - D_FC - p - ineq - 4}), since $\mathbf{\tilde{x}} = \{ \tilde{x}_1, \cdots, \tilde{x}_M \}$ is a probability mass function, $\displaystyle \sum_{m = 1}^M \hat{x}_m = 1$.

Therefore, Equation (\ref{Eqn: Diff - D_FC - p - ineq - 4}) reduces to
\begin{equation}
	\begin{array}{l}
		\displaystyle \frac{\partial D_{FC}}{\partial p} \ = \ \displaystyle T_1 - \frac{M \alpha}{1 - M \alpha p} D(\mathbf{\tilde{x}} || \mathbf{\tilde{y}}) + \frac{M \alpha}{1 - M \alpha p} 
		\\ \\ \qquad \qquad \qquad \displaystyle + \frac{M \alpha^2 p}{1 - M \alpha p}  \sum_{m = 1}^M \left[ \log \frac{\tilde{x}_m}{\tilde{y}_m} - \frac{\tilde{x}_m}{\tilde{y}_m} \right].
	\end{array}
	\label{Eqn: Diff - D_FC - p - ineq - 5}
\end{equation}
Note that the last term in the RHS of Equation (\ref{Eqn: Diff - D_FC - p - ineq - 5}),
$$\displaystyle \frac{M \alpha^2 p}{1 - M \alpha p}  \sum_{m = 1}^M \left[ \log \frac{\tilde{x}_m}{\tilde{y}_m} - \frac{\tilde{x}_m}{\tilde{y}_m} \right] = \frac{M \alpha p}{1 - M \alpha p} T_1.$$
In other words,
\begin{equation}
	\begin{array}{l}
		\displaystyle \frac{\partial D_{FC}}{\partial p} = \displaystyle \left( 1 + \frac{M \alpha p}{1 - M \alpha p}\right)T_1 - \frac{M \alpha}{1 - M \alpha p} D(\mathbf{\tilde{x}} || \mathbf{\tilde{y}}) 
		\\ \\ \qquad \qquad \qquad \qquad \qquad \displaystyle + \frac{M \alpha}{1 - M \alpha p}
		\\
		\\
		\quad = \displaystyle \frac{1}{1 - M \alpha p} T_1 - \frac{M \alpha}{1 - M \alpha p} D(\mathbf{\tilde{x}} || \mathbf{\tilde{y}}) + \frac{M \alpha}{1 - M \alpha p}.
	\end{array}
	\label{Eqn: Diff - D_FC - p - ineq - 6}
\end{equation}

Rearranging the terms in Equation (\ref{Eqn: Diff - D_FC - p - ineq - 6}) and expanding $T_1$, we have
\begin{equation}
	\begin{array}{l}
		\displaystyle \frac{\partial D_{FC}}{\partial p} \ = \ \displaystyle - \frac{M \alpha}{1 - M \alpha p} D(\mathbf{\tilde{x}} || \mathbf{\tilde{y}}) + \frac{M \alpha}{1 - M \alpha p} 
		\\ \\ \qquad \qquad \qquad \displaystyle + \frac{\alpha}{1 - M \alpha p} \sum_{m = 1}^M \left[ \log \frac{\tilde{x}_m}{\tilde{y}_m} - \frac{\tilde{x}_m}{\tilde{y}_m} \right]
		\\
		\\
		\quad = \ \displaystyle - \frac{M \alpha}{1 - M \alpha p} D(\mathbf{\tilde{x}} || \mathbf{\tilde{y}}) 
		\\ \\ \qquad \qquad \qquad \displaystyle + \frac{\alpha}{1 - M \alpha p} \sum_{m = 1}^M \left[ \log \frac{\tilde{x}_m}{\tilde{y}_m} - \left( \frac{\tilde{x}_m}{\tilde{y}_m} - 1\right) \right].
	\end{array}
	\label{Eqn: Diff - D_FC - p - ineq - 7}
\end{equation}

Since $\log x \leq x - 1$ for all $x$, we find that the second term in the RHS of Equation (\ref{Eqn: Diff - D_FC - p - ineq - 2}) is negative. 
Therefore, we have
\begin{equation}
	\displaystyle \frac{\partial D_{FC}}{\partial p} \leq 0.
\end{equation}
Since $D_{FC}$ is a non-increasing function of $p$, the optimal $p$, $p^*$, takes the maximum value $1/(M-1)$.

\section{Proof for Theorem \ref{Thrm: Optimal Byzantine Attack - estimation - constrained - symmetric} \label{App: Theorem - Optimal Byzantine Attack - estimation - constrained - symmetric}}
%
For the sake of notational simplicity, we let $z_m = P(u = m|\theta)$. Similarly, $\tilde{z}_m = P(v = m|\theta)$. Using this notation in Equation (\ref{Eqn: J_FC}), we have
\begin{equation}
	\begin{array}{lcl}
		\displaystyle J_{FC} & = & \displaystyle  \sum_{m = 1}^M P(v = m | \theta) \left( \frac{\partial \log P(v = m | \theta)}{\partial \theta} \right)^2
		\\
		\\
		& = & \displaystyle \sum_{m = 1}^M \tilde{z}_m \left( \frac{\partial \log \tilde{z}_m}{\partial \theta} \right)^2
		\\
		\\
		& = & \displaystyle (1 - M \alpha p)^2 \sum_{m = 1}^M \frac{1}{\tilde{z}_m} \left( \frac{\partial z_m}{\partial \theta} \right)^2.
	\end{array}
	\label{Eqn: Diff - J_FC - p}
\end{equation}

Partially differentiating $J_{FC}$ with respect to $p$, we have 
\begin{equation}
	\begin{array}{l}
		\displaystyle \frac{\partial J_{FC}}{\partial p} \ = \ \displaystyle 2 (1 - M \alpha p) (-M \alpha) \sum_{m = 1}^M \frac{1}{\tilde{z}_m} \left( \frac{\partial z_m}{\partial \theta} \right)^2 
		\\ \\ \qquad \qquad \displaystyle + (1 - M \alpha p)^2 \sum_{m = 1}^M \left( - \frac{1}{\tilde{z}_m^2} \right) (\alpha - M \alpha z_m) \left( \frac{\partial z_m}{\partial \theta} \right)^2
		\\
		\\
		\quad = \displaystyle -(1 - M \alpha p) \left[ 2 M \alpha \sum_{m = 1}^M \tilde{z}_m \left( \frac{1}{\tilde{z}_m} \frac{\partial z_m}{\partial \theta} \right)^2 \right.
		\\ \\ \qquad \qquad \qquad \qquad \displaystyle \ + (1 - M \alpha p) \sum_{m = 1}^M \alpha\left( \frac{1}{\tilde{z}_m} \frac{\partial z_m}{\partial \theta} \right)^2 
		\\ \\ \qquad \qquad \qquad \qquad \quad \left. \displaystyle - (1 - M \alpha p) \sum_{m = 1}^M M \alpha z_m \left( \frac{1}{\tilde{z}_m} \frac{\partial z_m}{\partial \theta} \right)^2 \right]
		\\
		\\
		\quad = \displaystyle  -(1 - M \alpha p) \left[ \alpha (1 - M \alpha p) \sum_{m = 1}^M \left( \frac{1}{\tilde{z}_m} \frac{\partial z_m}{\partial \theta} \right)^2 \right.
		\\ \\ \qquad \qquad \qquad \qquad \quad \left. \displaystyle + M \alpha (1 + M \alpha p) \sum_{m = 1}^M z_m \left( \frac{1}{\tilde{z}_m} \frac{\partial z_m}{\partial \theta} \right)^2 \right].
	\end{array}
	\label{Eqn: Diff - J_FC - p - ineq}
\end{equation}
In Equation (\ref{Eqn: Diff - J_FC - p - ineq}), we have a negative term multiplied by a non-negative term, and hence we have
\begin{equation}
	\displaystyle \frac{\partial J_{FC}}{\partial p} \leq 0.
\end{equation}
Since $J_{FC}$ is a non-increasing function of $p$, $p^* = \displaystyle \frac{1}{M-1}$, being the maximum value, is the optimal solution to Problem \ref{Prob. Form: constrained attack - estimation - symmetric}.

\section*{Acknowledgement}
This work was supported in part by AFOSR under Grants FA9550-10-1-0458, FA9550-10-1-0263, FA
9550-10-C-0179 and by CASE at Syracuse University and National Science Council of Taiwan, under grants NSC 99-2221-E-011-158 -MY3, NSC 101-2221-E-011-069 -MY3. Han's work was completed during his visit to Syracuse University from 2012 to 2013.


\bibliographystyle{IEEEtran}    
\bibliography{IEEEabrv,references}


\end{document}